\newif\ifdraft \drafttrue
\newif\iffull \fulltrue

\makeatletter \@input{tex.flags} \makeatother

\documentclass[11pt]{article}

\usepackage{algorithm}
\usepackage[noend]{algorithmic}

\usepackage{fullpage}

\usepackage{amsmath, amssymb, amsthm}
\usepackage{microtype}
\usepackage{mathtools}
\usepackage{verbatim}
\usepackage{color}
\usepackage{times}
\usepackage{xcolor}
\usepackage{tikz}
\usepackage{multicol}

\definecolor{DarkGreen}{rgb}{0.2,0.6,0.2}
\definecolor{DarkRed}{rgb}{0.6,0.2,0.2}
\definecolor{DarkBlue}{rgb}{0.2,0.2,0.6}
\usepackage[]{hyperref}
\hypersetup{
    unicode=false,          
    pdftoolbar=true,        
    pdfmenubar=true,        
    pdffitwindow=false,      
    pdftitle={},    
    pdfauthor={}
    pdfsubject={},   
    pdfnewwindow=true,      
    pdfkeywords={keywords}, 
    colorlinks=true,       
    linkcolor=DarkBlue,          
    citecolor=DarkBlue,        
    filecolor=DarkRed,      
    urlcolor=DarkBlue,          
}
\usepackage{xspace}
\usepackage{cleveref}
\usepackage[numbers]{natbib}

\newcommand{\sw}[1]{\ifdraft \textcolor{blue}{[Steven: #1]}\else \fi}
\newcommand{\ju}[1]{\ifdraft \textcolor{red}{[Jon: #1]}\else \fi}
\newcommand{\ar}[1]{\ifdraft \textcolor{brown}{[Aaron: #1]}\else \fi}

\newcommand\RR{\mathbb{R}}
\newcommand\RRP{\mathbb{R}_{+}}
\newcommand\RRp{\mathbb{R}_{>0}}
\newcommand\cA{\mathcal{A}}
\newcommand\cB{\mathcal{B}}

\newcommand\cF{\mathcal{F}}

\newcommand\cP{\mathcal{P}}
\newcommand\pp{P^*}
\newcommand\cG{\mathcal{G}}

\newcommand\cR{\mathcal{R}}
\newcommand\cX{\mathcal{X}}
\newcommand\cD{\mathcal{D}}

\newcommand\cN{\mathcal{N}}
\newcommand\cS{\mathcal{S}}

\newcommand\cL{\mathcal{L}}
\newcommand{\cI}{\mathcal{I}}

\newcommand{\bl}{^\bullet}
\newcommand{\lambdac}{\lambda_{\mathrm{cost}}}
\newcommand{\lambdav}{\lambda_{\mathrm{val}}}

\DeclareMathOperator{\poly}{poly}
\DeclareMathOperator{\polylog}{polylog}

\newcommand\set[1]{\left\{#1\right\}} 
\newcommand{\from}{:}

\newcommand{\proj}[1]{\Pi_{#1}}

\newcommand{\ex}[1]{\mathbb{E}\left[#1\right]}
\DeclareMathOperator*{\Expectation}{\mathbb{E}}
\newcommand{\Ex}[2]{\Expectation_{#1}\left[#2\right]}

\newcommand{\lp}{\mbox{{\bf LearnPrice}}\xspace}
\newcommand{\lpn}{\mbox{{\bf LearnPriceN}}\xspace}
\newcommand{\lpe}{\mbox{{\bf LearnPE}}\xspace}
\newcommand{\lpf}{\mbox{{\bf LearnTE}}\xspace}
\newcommand{\LL}{\mbox{{\bf LearnLead}}\xspace}
\newcommand{\zero}{\mbox{{\bf ZOO}}\xspace}
\newcommand{\op}{\mbox{{\bf Opro}}\xspace}
\newcommand{\opn}{\mbox{{\bf OproN}}\xspace}
\newcommand{\lop}{\mbox{{\bf LearnOpt}}\xspace}
\newcommand{\ellip}{\mbox{{\bf ellip}}\xspace}
\newcommand{\eps}{\varepsilon}
\def\epsilon{\varepsilon}

\DeclareMathOperator{\OPT}{OPT}
\newcommand{\RO}{\textrm{R-OPT}}
\renewcommand{\hat}{\widehat}

\newcommand{\mess}{\left(\lambdav\right)^{1/\beta}\,\left( \frac{4}{\eps^2\sigma} \right)^{(1-\beta)/\beta}}

\DeclareMathOperator*{\argmin}{\mathrm{argmin}}
\DeclareMathOperator*{\argmax}{\mathrm{argmax}}

\newcommand{\INDSTATE}[1][1]{\STATE\hspace{#1\algorithmicindent}}

\newtheorem{theorem}{Theorem}
\newtheorem{lemma}[theorem]{Lemma}

\newtheorem{claim}[theorem]{Claim}
\newtheorem{remark}[theorem]{Remark}

\newtheorem{definition}[theorem]{Definition}
\newtheorem*{theorem*}{Theorem}
\newtheorem{assumption}{Assumption}[section]
\newtheorem{example}{Example}

\title{Watch and Learn: \\ Optimizing from Revealed Preferences Feedback}

\author{
Aaron Roth\thanks{University of Pennsylvania Department of Computer
and Information Sciences. Partially supported by an NSF CAREER award,
NSF Grants CCF-1101389 and CNS-1065060, and a Google Focused Research
Award. \href{mailto:aaroth@cis.upenn.edu}{aaroth@cis.upenn.edu}.} \and
Jonathan Ullman\thanks{Northeastern University College of Computer and Information Science.  Most of this work was done while the author was supported by a Junior Fellowship from the Simons Society of Fellows. \href{mailto:jullman@ccs.neu.edu}{jullman@ccs.neu.edu}.} \and Zhiwei Steven Wu\thanks{University of
Pennsylvania Department of Computer and Information
Sciences. \href{mailto:wuzhiwei@cis.upenn.edu}{wuzhiwei@cis.upenn.edu}.}
}

\begin{document}

\maketitle

\pagenumbering{gobble}
\begin{abstract}
A \emph{Stackelberg game} is played between a \emph{leader} and a \emph{follower}.  The leader first chooses an action, then the follower plays his best response.  The goal of the leader is to pick the action that will maximize his payoff given the follower's best response. In this paper we present an approach to solving for the leader's optimal strategy in certain Stackelberg games where the follower's utility function (and thus the subsequent best response of the follower) is \emph{unknown}.

Stackelberg games capture, for example, the following interaction between a producer and a consumer.  The producer chooses the \emph{prices} of the goods he produces, and then a consumer chooses to buy a utility maximizing bundle of goods. The goal of the seller here is to set prices to maximize his profit---his revenue, minus the production cost of the purchased bundle. It is quite natural that the seller in this example should not know the buyer's utility function. However, he does have access to \emph{revealed preference} feedback---he can set prices, and then observe the purchased bundle and his own profit. We give algorithms for efficiently solving, in terms of both computational and query complexity, a broad class of Stackelberg games in which the follower's utility function is unknown, using only ``revealed preference'' access to it. This class includes in particular the profit maximization problem, as well as the optimal tolling problem in nonatomic congestion games, when the latency functions are unknown. Surprisingly, we are able to solve these problems even though the optimization problems are non-convex in the leader's actions.
\end{abstract}
\vfill
\newpage

\iffull
\tableofcontents
\vfill
\newpage
\fi

\pagenumbering{arabic}
\section{Introduction}
Consider the following two natural problems:
\begin{enumerate}
\item \textbf{Profit Maximization via Revealed Preferences}: A
  retailer, who sells $d$ goods, repeatedly interacts with a buyer. In
  each interaction, the retailer decides how to price the $d$ goods by
  choosing $p \in \mathbb{R}^d_+$, and in response, the buyer
  purchases the bundle $x \in \mathbb{R}^d_+$ that maximizes her
  utility $v(x) - \langle x, p \rangle$, where $v$ is an unknown
  concave valuation function. The retailer observes the bundle
  purchased, and therefore his profit, which is $\langle x, p \rangle
  - c(x)$, where $c$ is an unknown convex cost function. The retailer
  would like to set prices that maximize his profit after only a
  polynomial number of interactions with the buyer.

\item \textbf{Optimal Tolling via Revealed Behavior}: A municipal
  authority administers $m$ roads that form a network $G =
  (V,E)$. Each road $e \in E$ of the network has an unknown latency
  function $\ell_e:\mathbb{R}_+\rightarrow \mathbb{R}_+$ which
  determines the time it takes to traverse the road given a level of
  congestion. The authority has the power to set constant tolls
  $\tau_e \in \mathbb{R}_+$ on the roads in an attempt to manipulate
  traffic flow. In rounds, the authority sets tolls, and then observes
  the \emph{Nash equilibrium flow} induced by the non-atomic network
  congestion game defined by the unknown latency functions and the
  tolls, together with the social cost (average total latency) of the
  flow. The authority would like to set tolls that minimize the social
  cost after only a polynomial number of rounds.
\end{enumerate}

Although these problems are quite different, they share at least one
important feature---the retailer and the municipal authority each wish
to optimize an \emph{unknown} objective function given only query
access to it.  That is, they have the power to choose some set of
prices or tolls, and then observe the value of their objective
function that results from that choice. This kind of problem
(alternately called \emph{bandit} or \emph{zeroth order} optimization)
is well-studied, and is well understood in cases in which the unknown
objective being maximized (resp.~minimized) is concave
(resp.~convex). Unfortunately, the two problems posed above share
another important feature---when posed as bandit optimization
problems, the objective function being maximized (resp.~minimized) is
generally not concave (resp.~convex).  For the profit maximization
problem, even simple instances lead to a non concave objective
function.

\begin{example}
Consider a setting with one good ($d = 1$).  The buyer's valuation function $v(x) = \sqrt{x}$, and the retailer's cost function is $c(x) = x$.  The buyer's utility for buying $x$ units at price $p$ is $\sqrt{x} - x \cdot p$.  Thus, if the price is $p$, a utility-maximizing buyer will purchase $x^*(p) = \tfrac{1}{4 p^2}$ units.  The profit of the retailer is then
\iffull \[ \fi
\mathrm{Profit}(p) = p\cdot x^*(p) - c(x^*(p)) = \frac{1}{4 p} - \frac{1}{4 p^2}.
\iffull \] \fi
Unfortunately, this profit function is not concave.
\end{example}

Since the retailer's profit function is not concave in the prices, it
cannot be optimized efficiently using generic methods for concave
maximization.  This phenomenon persists into higher dimensions, where
it is not clear how to efficiently maximize the non-concave objective.
The welfare objective in the tolling problem is also non-convex in the
tolls.  \iffull We give an example in Appendix
\ref{sec:routingexample}.  \fi

Surprisingly, despite this non-convexity, we show that both of these
problems can be solved efficiently subject to certain mild conditions.
More generally, we show how to solve a large family of
\emph{Stackelberg games} in which the utility function of the
``follower'' is unknown.  A \emph{Stackelberg game} is played by a
\emph{leader} and a \emph{follower}. The leader moves first and
commits to an action (e.g., setting prices or tolls as in our
examples), and then the follower \emph{best responds}, playing the
action that maximizes her utility given the leader's action. The
leader's problem is to find the action that will optimize his
objective (e.g., maximizing profit, or minimizing social cost as in
our examples) after the follower best responds to this action.

Traditionally, Stackelberg games are solved assuming that the leader
knows the follower's utility function, and thus his own utility
function.  But this assumption is very strong, and in many realistic
settings the follower's utility function will be unknown.  Our results
give general conditions---and several natural examples---under which
the problem of computing an optimal Stackelberg equilibrium can be
solved efficiently with only revealed preferences feedback to the
follower's utility function.

\iffull For clarity of exposition, we first work out our solution in
detail for the special case of profit maximization from revealed
preferences.We then derive and state our general theorem for optimally
solving a class of Stackelberg games where the follower's utility is
unknown.  Finally, we show how to apply the general theorem to other
problems, including the optimal tolling problem mentioned above and a
natural principal-agent problem.  \fi
\subsection{Our Results and Techniques}
The main challenge in solving our class of Stackelberg games is that
for many natural examples, the leader's objective function is not
concave when written as a function of his own action.  For instance,
in our example, the retailer's profit is not concave as a function of
the price he sets.  Our first key ingredient is to show that in many
natural settings, \emph{the leader's objective is concave when written
  as a function of the follower's action}.

Consider again the retailer's profit maximization problem.  Recall that if the buyer's valuation function $v(x) = \sqrt{x}$, then when she faces a price $p$, she will buy the bundle $x^*(p) = 1/4p^2$.  In this simple case, we can see that setting a price of $p^*(x) = 1/2\sqrt{x}$ will induce the buyer to purchase $x$ units.
In principle, we can now write the retailer's profit function as a function of the bundle $x$.  In our example, the retailer's cost function is simply $c(x) = x$.  So,
\iffull \[ \fi
\mathrm{Profit}(x) = p^*(x) \cdot x - c(x) = \frac{\sqrt{x}}{2} - x.
\iffull \] \fi

Written in terms of $x$, the profit function is concave! As we show,
this phenomenon continues in higher dimensions, for arbitrary convex
cost functions $c$ and for a wide class of concave valuation functions
satisfying certain technical conditions, including the well studied
families of CES and Cobb-Douglas utility functions.

Thus, if the retailer had access to an oracle for the concave function $\mathrm{Profit}(x)$, we could use an algorithm for bandit concave optimization to maximize the retailer's profit. Unfortunately, the retailer does not directly get to choose the bundle purchased by the buyer and observe the profit for that bundle: he can only set \emph{prices} and observe the buyer's chosen bundle $x^*(p)$ at those prices, and the resulting profit $\mathrm{Profit}(x^*(p))$.

Nevertheless, we have reduced the retailer's problem to a possibly simpler one.  In order to find the profit maximizing prices, it suffices to give an algorithm which simulates access to an oracle for $\mathrm{Profit}(x)$ given only the retailer's query access to $x^*(p)$ and $\mathrm{Profit}(x^*(p))$.  Specifically, if for a given bundle $x$, the retailer could find prices $p$ such that the buyer's chosen bundle $x^*(p) = x$, then he could simulate access to $\mathrm{Profit}(x)$ by setting prices $p$ and receiving $\mathrm{Profit}(x^*(p)) = \mathrm{Profit}(x)$.

Our next key ingredient is a ``t\^{a}tonnement-like'' procedure that
efficiently finds prices that approximately induce a target bundle $x$
given only access to $x^*(p)$, provided that the buyer's valuation
function is H\"{o}lder continuous\sw{changed} and strongly concave on
the set of feasible bundles.  Specifically, given a target bundle $x$,
our procedure finds prices $p$ such that $| \mathrm{Profit}(x^*(p)) -
\mathrm{Profit}(x) | \leq \eps$.  Thus, we can use our procedure to
simulate approximate access to the function $\mathrm{Profit}(x)$.  Our
procedure requires only $\poly(d, 1/\eps)$ queries to $x^*(p)$.  Using
recent algorithms for bandit optimization due to Belloni et
al.~\cite{BLNR15}, we can maximize the retailer's profits efficiently
even with only approximate access to $\mathrm{Profit}(x)$. When our
algorithms receive \emph{noiseless} feedback, we can improve the
dependence on the approximation parameter $\epsilon$ to be only
$\poly(\log 1/\eps)$.

A similar approach can be used to solve the optimal tolling problem assuming the unknown latency functions are convex and strictly increasing.  As in the preceding example, the municipal authority's objective function (social cost) is not convex in the tolls, but \emph{is} convex in the induced flow. Whenever the latency function are strictly increasing, the \emph{potential function} of the routing game is strongly convex, and so we can use our t\^{a}tonnement procedure to find tolls that induce target flows at equilibrium.

Our results for maximizing profits and optimizing tolls follow from a more general method that allows the leader in a large class of continuous action Stackelberg game to iteratively and efficiently maximize his objective function while only observing the follower's response.  The class requires the following conditions:
\iffull
\begin{enumerate}
\item The follower's utility function is strongly concave in her own actions and linear in the leader's actions.
\item The leader's objective function is concave when written as a function of the follower's actions.\footnote{When the leader and follower are instead trying to minimize a cost function, replace ``concave'' with ``convex'' in the above.}
\end{enumerate}
\fi

Finally, we show that our techniques are tolerant to two different kinds of noise. Our techniques work even if the follower only approximately maximizes his utility function, which corresponds to bounded, but adversarially chosen noise -- and also if unbounded, but well behaved (i.e. zero mean and bounded variance) noise is introduced into the system. To illustrate this noise tolerance, we show how to solve a simple $d$-dimensional principal-agent problem, in which the principal contracts for the production of $d$ types of goods that are produced as a stochastic function of the agent's actions.

\subsection{Related Work}
There is a very large literature in operations research on solving
so-called ``bilevel programming'' problems, which are closely related
to Stackelberg games.  Similar to a Stackelberg game, the variables in
a bilevel programming problem are partitioned into two ``levels.''
The second-level variables are constrained to be the optimal solution
to some problem defined by the first-level variables.  See
\cite{CMS05} for a survey of the bilevel programming literature.
Unlike our work, this literature does not focus substantially on
computational issues (many of the algorithms are not polynomial
time). \cite{KCP10} show that optimally solving certain discrete
Stackelberg games is NP-hard.  Even ignoring computational efficiency,
this literature assumes knowledge of the objective function of the
``follower.'' Our work departs significantly from this literature by
assuming that the leader has no knowledge of the follower's utility
function.

There are two other works that we are aware of that consider solving
Stackelberg games when the follower's utility function is
unknown. Letchford, Conitzer, and Munagala \cite{LCM09} give
algorithms for learning optimal leader strategies with a number of
queries that is polynomial in the number of pure strategies of the
leader. In our setting, the leader has a continuous and high dimensional action space, and
so the results of \cite{LCM09} do not apply. Blum, Haghtalab, and
Procaccia \cite{BHP14} consider the problem of learning optimal
strategies for the leader in a class of \emph{security games}. They
exploit the structure of security games to learn optimal strategies
for the leader in a number of queries that is polynomial in the
representation size of the game (despite the fact that the number of
pure strategies is exponential). The algorithm of \cite{BHP14} is not
computationally efficient -- indeed, the problem they are solving is
NP-hard. Neither of these techniques apply to our setting -- and
despite the fact that in our setting the leader has a continuous action
space (which is exponentially large even under discretization), we are able to give an algorithm with both polynomial query
complexity and polynomial running time.

There is also a body of related work related to our main example of
profit maximization.  Specifically, there is a recent line of work on
\emph{learning to predict from revealed preferences}
(\cite{BV06,ZR12,BDMRV14}). In this line, the goal is to \emph{predict} buyer behavior, rather than to optimize seller prices. Following these works, Amin et
al.~\cite{ACDKR15} considered how to find profit maximizing pricing
from revealed preferences in the special case in which the buyer has a
linear utility function and a fixed budget. The technique of
\cite{ACDKR15} is quite specialized to linear utility functions, and
does not easily extend to more general utility functions in the profit
maximization problem, and not to Stackelberg games in general. ``Revealed preferences'' queries are quite similar to \emph{demand queries} (see e.g. \cite{BN09}). Demand queries are known to be sufficient to find welfare optimal allocations, and more generally, to be able to solve separable convex programs whose objective is social welfare. In contrast, our optimization problem is non-convex (and so the typical methodology by which demand queries are used does not apply), and our objective is not welfare.

The profit maximization application can be viewed as a dynamic pricing problem in which the seller has no knowledge of the buyers utilities. Babaioff et al. \cite{BDKS15} study a version of this problem that is incomparable to our setting. On the one hand, \cite{BDKS15} allow for \emph{distributions} over buyers. On the other hand, \cite{BDKS15} is limited to selling a single type of good, whereas our algorithms apply to selling bundles of many types of goods. There is also work related to our optimal tolling problem.  In an
elegant paper, Bhaskar et al.~\cite{BLSS14} study how one can
iteratively find tolls such that a particular target flow is an
equilibrium of a non-atomic routing game where the latency functions
are unknown, which is a sub-problem we also need to solve in the routing application.  Their technique is specialized to routing games, and requires that the unknown latency functions have a known simple functional form (linear or low-degree convex polynomial). In contrast, our technique works quite generally, and in the special case of routing games, does not require the latency functions to satisfy any known functional form (or even be convex). Our technique can also be implemented in a noise tolerant way, although at the expense of having a polynomial dependence on the approximation parameter, rather than a polylogarithmic dependence (in the absence of noise, our method can also be implemented to depend only polylogarithmically on the approximation parameter.)


Finally, our work is related in motivation to a recent line of work
designed to study the \emph{sample complexity of auctions}
\cite{BBHM08,CR14,HMR14,DHN14,CHN14,BMM15,MR15}. In this line of work, like
in our work, the goal is to optimize an objective in a game theoretic
setting when the designer has no direct knowledge of participant's
utility functions.

\section{Preliminaries}
We will denote the set of non-negative real numbers by $\RRP = \set{x
  \in \RR \mid x \geq 0}$ and the set of positive real numbers by
$\RRp = \set{x \in \RR \mid x>0}$.  For a set $C \subseteq \RR^d$ and a norm $\| \cdot
\|$, we will use $\| C \| = \sup_{x \in C} \|x\|$ to denote the
diameter of $C$ with respect to the norm $\| \cdot \|$.  \iffull When
the norm is unspecified, $\| \cdot \|$ will denote the Euclidean norm
$\| \cdot \|_2$.  \fi

An important concept we use is the \emph{interior} of a set. In the
following, we will use $B_{u}$ to denote the unit ball centered at $u$
for any $u \in \RR^d$.

\begin{definition}
For any $\delta >0$ and any set $C\subseteq \RR^d$, the
$\delta$-interior $\mathrm{Int}_{C, \delta}$ of $C$ is a subset of $C$
such that a point $x$ is in the $\delta$-interior $\mathrm{Int}_{C,
  \delta}$ of $C$ if the ball of radius $\delta$ centered at $x$ is
contained in $C$, that is: \iffull
\[
x + \delta B_0 = \{x + \delta y \mid \|y\| \leq 1\}\subseteq C.
\]
\fi The interior $\mathrm{Int}_C$ of $C$ is a subset of $C$ such
that a point $x$ is in $\mathrm{Int}_C$ if there exists some $\delta'
> 0$ such that $x$ is in $\mathrm{Int}_{C, \delta'}$.
\end{definition}

We will also make use of the notions of H\"{o}lder continuity and
Lipschitzness.
\begin{definition}
A function $f\colon C\rightarrow \RR$ is $(\lambda,
\beta)$-H\"{o}lder continuous for some $\lambda, \beta\geq 0$ if
for any $x, y\in C$,
\[
|f(x) - f(y)| \leq \lambda \|x - y\|^\beta.
\]
A function $f$ is $\lambda$-Lipschitz if it is $(\lambda,
1)$-H\"{o}lder continuous.
\end{definition}

\iffull
\subsection{Projected Subgradient Descent}
A key ingredient in our algorithms is the ability to minimize a convex
function (or maximize a concave function), given access only to the
subgradients of the function (i.e.~with a so-called ``first-order''
method). For concreteness, in this paper we do so using the projected
sub gradient descent algorithm. This algorithm has the property that
it is~\emph{noise-tolerant}, which is important in some of our
applications. However, we note that any other noise-tolerant
first-order method could be used in place of gradient descent to
obtain qualitatively similar results. In fact, we show in the appendix
that for applications that do not require noise tolerance, we can use
the Ellipsoid algorithm, which obtains an exponentially better
dependence on the approximation parameter. Because we strive for
generality, in the body of the paper we restrict attention to gradient
descent.

Let $C \subseteq \RR^d$ be a compact and convex set that is contained in a
Euclidean ball of radius $R$, centered at some point $x_1 \in \RR^d$.  Let $c \from \RR^d \to \RR$ be a convex
``loss function.'' Assume that $c$ is also $\lambda$-Lipschitz----that is, $|c(x) - c(y)| \leq \lambda\|x-y\|_2.$
Let $\Pi_C$ denote the projection operator onto $C,$
\[
\Pi_C (x) = \argmin_{y\in C} \|x-y\|.
\]
Projected subgradient descent is an iterative algorithm that starts at
$x_1\in C$ and iterates the following equations
\begin{align*}
  y_{t+1} &= x_t -\eta\, g_t, \mbox{ where }g_t\in \partial c(x_t)\\
  x_{t+1} &= \Pi_\cX (y_{t+1})
\end{align*}

The algorithm has the following guarantee.
\begin{theorem}
\label{thm:gd}
  The projected subgradient descent algorithm with $\eta = \frac{R}{\lambda\sqrt{T}}$ satisfies
  \[
  c\left(\frac{1}{T} \sum_{t=1}^T x_s \right) \leq \min_{y \in C} c(y) + \frac{R\lambda}{\sqrt{T}}
  \]
Alternatively, the algorithm finds a solution within $\eps$ of optimal after $T = (R\lambda/\eps)^2$ steps.
\end{theorem}
\fi

\subsection{Strong Convexity}
We will make essential use of \emph{strong convexity/concavity} of certain functions.
\begin{definition}
Let $\phi\colon C \rightarrow \RR$ be a function defined over a convex
set $C \subseteq \RR^d$.  We say $\phi$ is \emph{$\sigma$-strongly
  convex} if for every $x, y \in C$, \iffull \[ \fi \phi(y)
\geq \phi(x) + \langle \nabla \phi(x), y - x \rangle + \frac{\sigma}{2}
\cdot \|y - x\|^2_2.  \iffull \] \fi We say $\phi$ is
\emph{$\sigma$-strongly concave} if $(-\phi)$ is $\sigma$-strongly
convex.
\end{definition}

An extremely useful property of strongly convex functions is that any point in the domain that is close to the minimum in objective value is also close to the minimum in Euclidean distance.
\begin{lemma}
\label{lem:sconvex}
Let $\phi\colon C \rightarrow \RR$ be a $\sigma$-strongly convex function, and let $x^* = \argmin_{x \in C} \phi(x)$ be the minimizer of $\phi$.  Then, for any $x\in C$,
\iffull \[ \fi
\|x - x^*\|^2_2 \leq \frac{2}{\sigma} \cdot (\phi(x) - \phi(x^*)).
\iffull \] \fi
\iffull
Similarly, if $\phi$ is $\sigma$-strongly concave, and $x^* = \argmax_{x \in C} \phi(x)$, then for any $x \in C$,
\iffull \[ \else $ \fi
\|x - x^*\|_2^2 \leq \frac{2}{\sigma} \cdot (\phi(x^*) - \phi(x)).
\iffull \] \else $ \fi\fi
\end{lemma}
\subsection{Tools for Zeroth-Order Optimization}

We briefly discuss a useful tool for noisy zeroth-order
optimization (also known as \emph{bandit} optimization) by~\cite{BLNR15}, which will be used as blackbox
algorithm in our framework. The important feature we require, satisfied by the algorithm from \cite{BLNR15} is that the optimization procedure be able to tolerate a small amount of adversarial noise. 

\begin{definition}
  Let $C$ be a convex set in $\RR^d$. We say that $C$ is
  well-rounded if there exist $r , R > 0$ such that $\cB^d_2(r)
  \subseteq C \subseteq \cB_2^d(R)$ and $R/r \leq O(\sqrt{d})$, where
  $\cB_2^d(\gamma)$ denotes an $\ell_2$ ball of radius $\gamma$ in $\RR^d$.
\end{definition}

Let $C$ be a well-rounded convex set in $\RR^d$ and $F, f\colon
\RR^d \to \RR$ be functions such that $f$ is convex and $F$ satisfies
\begin{equation}
\sup_{x\in C} |F(x) - f(x)| \leq \eps/d,
\label{eq:fF}
\end{equation}
for some $\eps > 0$. The function $F$ can be seen as an oracle that
gives a noisy evaluation of $f$ at any point in $C$. Belloni et
al.~\cite{BLNR15} give an algorithm that finds a point $x\in C$ that approximately optimizes
the convex function $f$ and only uses function evaluations of $F$ at points in $x \in C$.  The set $C$
only needs to be specified via a membership oracle that decides if a point $x$ is in $C$ or not.

\begin{lemma}[\cite{BLNR15}, Corollary 1]
\label{lem:zero}
  Let $C$ be a well-rounded set in $\RR^d$ and $f$ and $F$ be functions that
   satisfy~\Cref{eq:fF}. There is an algorithm $\zero(\eps, C)$ (short for \emph{zeroth-order optimization}) that makes $\tilde{O}(d^{4.5})$ calls\footnote{The notation $\tilde O(\cdot)$ hides the logarithmic dependence on $d$ and $1/\eps$.} to $F$ and returns
  a point $x\in C$ such that
  \iffull \[ \fi
  \Expectation[ f(x)] \leq \min_{y\in C} f(y) + \eps.
  \iffull \] \fi
\end{lemma}

\iffull
Naturally, the algorithm can also be used to approximately maximize a concave function.
\else \fi

\section{Profit Maximization From Revealed Preferences}
\label{sec:reveal}
\subsection{The Model and Problem Setup}
Consider the problem of maximizing profit from revealed
preferences. In this problem, there is a \emph{producer}, who wants to
sell a bundle $x$ of $d$ divisible goods to a \emph{consumer}.
The bundles are vectors $x \in C$ where $C \subseteq \RRP^d$ is some
set of \emph{feasible bundles} that we assume is \emph{known} to both
the producer and consumer.

\begin{itemize}
\item The producer has an \emph{unknown cost function} $c: \RRP^d \to \RRP$.  He is allowed to set prices $p \in \RRP^{d}$ for each good, and receives profit
\iffull
\[
r(p) = \langle p, x^*(p) \rangle - c(x^*(p)),
\]
\fi
where $x^*(p)$ is the bundle of goods the consumer purchases at prices $p$.  His goal is to find the profit maximizing prices
\iffull
\[
p^* = \argmax_{p \in \RRP^{d}} r(p).
\]
\fi

\item The consumer has a \emph{valuation function} $v: \RRP^d \to
  \RRP$.  The valuation function is \emph{unknown to the producer}.
  The consumer has a \emph{quasi-linear} utility function $u(x, p) =
  v(x) - \langle p, x \rangle.$ Given prices $p$, the consumer will
  buy the bundle $x^*(p) \in C$ that maximizes her utility.  Thus,
\[
x^*(p) = \argmax_{x \in C} u(x, p) = \argmax_{x\in C} \left(v(x) - \langle x, p \rangle\right).
\]
We call $x^*(p)$ the \emph{induced bundle at prices $p$}.
\end{itemize}

In our model, in each time period $t$ the producer will choose prices
$p^t$ and can observe the resulting induced bundle $x^*(p^t)$ and
profit $r(p^t)$. We would like to design an algorithm so that after a
polynomial number of observations $T$, the profit $r(p^T)$ is nearly
as large as the optimal profit $r(p^*)$.

We will make several assumptions about the functions $c$ and $v$ and
the set $C$. We view these assumptions as comparatively mild:
\begin{assumption}[Set of Feasible Bundles] \label{ass:feasible}
The set of feasible bundles $C \subseteq \RRP^d$ is convex and
well-rounded.  It also contains the set $(0,1]^d \subseteq C$ (the
  consumer can simultaneously buy at least one unit of each good).
  Also, $\|C\|_2 \leq \gamma$ (e.g.~when $C = (0,1]^d$, we have
    $\gamma = \sqrt{d}$). Lastly, $C$ is \emph{downward closed}, in
    the sense that for any $x\in C$, there exists some $\delta\in(0,
    1)$ such that $\delta\, x \in C$ (the consumer can always choose buy less of each good).\sw{changed}
\end{assumption}

\iffull
\begin{assumption}[Producer's Cost Function] \label{ass:cost}
The producer's cost function $c \from \RRP^d \to \RR$ is convex and
Lipschitz-continuous.
\end{assumption}
\fi

\iffull
\begin{assumption}[Consumer's Valuation Function] \label{ass:valuation}
The consumer's valuation function $v \from \RRP^d \to \RR$ is
non-decreasing, H\"{o}lder-continuous, differentiable and strongly
concave over $C$. For any price vector $p\in\RRP^d$, the induced
bundle $x^*(p)= \argmax_{x\in C} u(x, p)$ is defined.
\end{assumption}
\fi Note that without the assumption that the consumer's valuation
function is concave and that the producer's cost function is convex,
even with full information, their corresponding optimization problems
would not be polynomial time solvable.  Our fourth assumption of
\emph{homogeneity} is more restrictive , but as we observe, is
satisfied by a wide range of economically meaningful valuation
functions including CES and Cobb-Douglas utilities.  Informally,
homogeneity is a scale-invariance condition --- changing the
\emph{units} by which quantities of goods are measured should have a
predictable multiplicative effect on the buyer valuation functions:

\begin{definition}
For $k \geq 0$, a function $v:\RR^d_+\rightarrow \RR_+$ is
\emph{homogeneous of degree $k$} if for every $x \in \RR^d$ and for
every $\sigma > 0$, \iffull $$v(\sigma x) = \sigma^k v(x).$$ \fi The function $v$ is simply
\emph{homogeneous} if it is homogeneous of degree $k$ for some $k \geq
0$.
\end{definition}

Our fourth assumption is simply that the buyer valuation function is
homogeneous of \emph{some} degree:
\begin{assumption}
\label{ass:concaverevenue}
  The consumer's valuation function $v$ is homogeneous.
\end{assumption}

\subsection{An Overview of Our Solution}


We present our solution in three main steps:
\begin{enumerate}
\item \label{cutholeinbox} First, we show that the profit function can
  be expressed as a concave function $r(x)$ of the consumer's induced
  bundle $x$, rather than as a (non-concave) function of the prices.
\item \label{putjunkinbox} Next, we show that for a given candidate
  bundle $x$, we can iteratively find prices $p$ such that $x \approx
  x^*(p)$.  That is, in each time period $s$ we can set prices $p^s$
  and observe the purchased bundle $x^*(p^s)$, and after a polynomial
  number of time periods $S$, we are guaranteed to find prices $p =
  p^S$ such that $x^*(p) \approx x$.  Once we have found such prices,
  we can observe the profit $r(x^*(p)) \approx r(x)$, which allows us
  to simulate query access to $r(x)$.
\item Finally, we use our simulated query access to $r(x)$ as feedback
  to a bandit  concave optimization algorithm, which iteratively
  queries \emph{bundles} $x$, and quickly converges to the profit
  maximizing bundle.
\end{enumerate}

\subsection{Expressing Profit as a Function of the Bundle}

First, we carry out Step~\ref{cutholeinbox} above and demonstrate how
to rewrite the profit function as a function of the bundle $x$, rather
than as a function of the prices $p$.  Note that for any given bundle
$x\in C$, there might be multiple price vectors that induce $x$. We
denote the set of price vectors that induce $x$ by: \iffull
$$\pp(x) = \{p\in \RR^d \mid x^*(p) = x\}.$$
\fi
We then define the profit of a bundle $x$ to be
$$r(x) = \max_{p \in \pp(x)} r(p) = \max_{p \in \pp(x)} \langle p, x
\rangle - c(x).$$

Observe that the profit maximizing price vector $p \in \pp(x)$ is the
price vector that maximizes \emph{revenue} $\langle p, x \rangle$,
since the \emph{cost} $c(x)$ depends only on $x$, and so is the same
for every $p \in \pp(x)$.  The following lemma characterizes the
revenue maximizing price vector that induces any fixed bundle $x \in
C$.
\begin{lemma}
\label{lem:best-price}
Let $\hat x\in C$ be a bundle, and $\pp(\hat x)$ be the set of price
vectors that induce bundle $\hat x$. Then the price vector $p = \nabla
v(\hat x)$ is the revenue maximizing price vector that induces $\hat
x$.  That is, $\nabla v(\hat x) \in \pp(\hat x)$ and for any price
vector $p'\in \pp(\hat x)$, $\langle p' , \hat x\rangle \leq \langle
\nabla v(\hat x) , \hat x\rangle$.
\end{lemma}
\iffull
\begin{proof}
Observe that for any $x\in C$ the gradient of the consumer's utility
$u(x,p) = v(x) - \langle p, x \rangle$ with respect to $x$ is ($\nabla
v - p$).  If the prices are $p = \nabla v(\hat x)$, then since $v$ is
concave and $\nabla v(\hat x) - p = \mathbf{0}$, $\hat x$ is a
maximizer of the consumer's utility function. Thus, we have
$x^*(\nabla v(\hat x)) = \hat x$, and so $\nabla v(\hat x) \in
\pp(\hat x)$.

Suppose that there exists another price vector $p'\in \pp(\hat x)$
such that $p'\neq \nabla v(\hat x)$.  Since the function $u(\cdot,
p')$ is concave in $x$ and $\hat x \in \arg \max_{x\in C} u(x, p')$,
we know that for any $x' \in C$
\[
\left\langle \nabla v(\hat x) - p', x' - \hat x\right\rangle \leq 0,
\]
otherwise there is a feasible ascent direction, which contradicts the
assumption that $\hat x$ maximizes $u(x, p')$. By~\Cref{ass:feasible},
we know there exists some $\delta < 1$ such that $\delta \hat x \in
C$. Now consider $x'= \delta \hat x$, then it follows that
\[
\left \langle \nabla v(\hat x) - p', \left(1 - \delta \right)\hat x \right\rangle = \left(1 - \delta \right) \left(\left
\langle \nabla v(\hat x), \hat x \right\rangle - \left\langle p' ,
\hat x\right\rangle \right)\geq 0.
\]
Therefore, $\langle p' , x\rangle \leq \langle \nabla
v(x) , x\rangle$, as desired.  This completes the proof.
\end{proof}
\fi

With this characterization of the revenue maximizing price vector, we
can then rewrite the profit as a function of $x$ in closed form for
any $x\in C$: \iffull
\begin{equation}\label{ref:prof}
r(x) = \left\langle \nabla v(x), x \right\rangle - c(x).
\end{equation}
\fi

Next, we show that $r(x)$ is a concave function of $x$ whenever the
valuation $v$ satisfies~\Cref{ass:valuation} (concavity and
differentiability) and~\Cref{ass:concaverevenue} (homogeneity).

\begin{theorem}
\label{thm:applyEuler}
If the consumer's valuation function $v$ is differentiable,
homogeneous, and concave over $C$, the producer's profit
function $r(x) = \left\langle \nabla v(x), x \right\rangle - c(x)$ is
concave over the domain $C$.
\end{theorem}

\iffull
To prove this result, we invoke Euler's theorem for homogeneous functions:
\begin{theorem}[Euler's Theorem for Homogeneous Functions]
\label{thm:euler}
Let $v:C \rightarrow \RR_+$ be continuous and
differentiable. Then $v$ is homogeneous of degree $k$ if and only if
\iffull
$$\langle \nabla v(x), x \rangle = k\cdot v(x).$$
\else
$\langle \nabla v(x), x \rangle = k\cdot v(x).$
\fi
\end{theorem}
\else
\fi

\iffull
\begin{proof}[Proof of Theorem \ref{thm:applyEuler}]
Recall that:
$$r(x) = \left\langle \nabla v(x), x \right\rangle - c(x)$$
By the assumption that $v$ is continuous, differentiable, and homogeneous of some degree $k \geq 0$, we have by Euler's theorem that
$$r(x) = k v(x) - c(x)$$
Because by assumption, $v(x)$ is concave, and $c(x)$ is convex, we conclude that $r(x)$ is concave.
\end{proof}
\fi

 Finally, we note that many important and well studied classes of
valuation functions satisfy our assumptions -- namely
differentiability, strong concavity and homogeneity. Two classes of
interest include\iffull\begin{itemize}
\item \textbf{Constant Elasticity of Substitution (CES).}  Valuation functions of the form:
\[
v(x) = \left( \sum_{i=1}^{d} \alpha_i x_i^{\rho} \right)^{\beta},
\]
where $\alpha_i > 0$ for every $i \in [d]$ and $\rho, \beta > 0$ such
that $\rho < 1$ and $\beta\rho < 1$.  These functions are known to be
differentiable, H\"{o}lder continuous and strongly concave over the
set $(0, H]^d$ (see~\Cref{sec:sconcave} for a proof). Observe that
  $v(\sigma x) = ( \sum_{i=1}^{d} \alpha_i (\sigma x_i)^\rho )^{\beta}
  = \sigma^{\rho\beta} ( \sum_{i=1}^{d} \alpha_i x_i^\rho )^{\beta} =
  \sigma^{\rho\beta} v(x)$, so these functions are homogeneous of
  degree $k = \rho\beta$.

\item \textbf{Cobb-Douglas.}  These are valuation functions of the form
\[
v(x) = \prod_{i=1}^{d} x_i^{\alpha_i},
\]
where $\alpha_i > 0$ for every $i \in [d]$ and $\sum_{i=1}^{d}
\alpha_i < 1$.  These functions are known to be differentiable,
H\"{o}lder continuous and strongly concave over the set $(0, H]^d$
(see~\Cref{sec:sconcave} for a proof).  Observe that $v(\sigma x) =
\prod_{i=1}^{d} (\sigma x_i)^{\alpha_i} = (\prod_{i=1}^{d}
\sigma^{\alpha_i}) (\prod_{i=1}^{d} x_i^{\alpha_i}) =
\sigma^{\sum_{i=1}^{d} \alpha_i} \cdot v(x)$, so these functions are
homogeneous of degree $k = \sum_{i=1}^{d} \alpha_i$.
\end{itemize}
\fi

\subsection{Converting Bundles to Prices}
\label{sec:conversion}

Next, we carry out Step~\ref{putjunkinbox} and show how to find prices
$\hat p$ to induce a given bundle $\hat x$. Specifically, the producer
has a target bundle $\hat x \in C$ in mind, and would
like to learn a price vector $\hat p \in \RRP^{d}$ such that the
induced bundle $x^*(\hat p)$ is ``close'' to $\hat x$.  That is,
\iffull
\[
\|\hat x - x^*(\hat p) \|_2 \leq \eps,
\]
\fi
for some $\eps > 0$.

Our solution will actually only allow us to produce a price vector
$\hat p$ such that $\hat x$ and $x^*(\hat p)$ are ``close in value.''
That is
\iffull
\[
|u(\hat x, \hat p) - u(x^*(\hat p), \hat p) | \leq \delta.
\]
\fi
However, by strong concavity of the valuation function, this will be
enough to guarantee that the actual bundle is close to the target
bundle. \iffull The following is just an elaboration of assumption
\ref{ass:valuation}:
\begin{assumption}[Quantitative version of~\Cref{ass:valuation}]
\label{ass:strongconcavity}
The valuation function $v$ is both
\begin{enumerate}
\item $(\lambdav, \beta)$-H\"{o}lder continuous over the domain $C$
  with respect to the $\ell_2$ norm---for all $x, x'\in C$, \iffull
\[
|v(x) - v(x')| \leq \lambdav \cdot \|x - x'\|_2^\beta,
\]
for some constants $\lambdav \geq 1$ and $\beta\in (0, 1]$, and
\else
$|v(x) - v(x')| \leq \lambdav \cdot \|x - x'\|_2$,
\fi
\item $\sigma$-strongly concave over the interior of $C$---for all $x,
  x'\in C$, \iffull
\[
v(x') \leq v(x) + \langle \nabla v(x), x' -
  x \rangle - (\sigma / 2) \cdot \| x - x' \|_2^2.
\]
\else
$v(x') \leq v(x) + \langle \nabla v(x), x' -
  x \rangle - (\sigma / 2) \cdot \| x - x' \|_2^2$.
\fi
\end{enumerate}
\end{assumption}
\fi

\begin{algorithm}[h]
  \caption{Learning the price vector to induce a target bundle: $\lp(\hat x, \eps)$}
 \label{alg:learnprice}
  \begin{algorithmic}

    \STATE{\textbf{Input:} A target bundle $\hat x\in C$, and target accuracy $\eps$}

    \INDSTATE{Initialize:
 restricted price space $\cP=\{p\in  \RR_+^{d} \mid \|p\|\leq \sqrt{d}L\}$ where
      \[
      L = \left(\lambdav\right)^{1/\beta}\,\left( \frac{4}{\eps^2\sigma} \right)^{(1-\beta)/\beta} \qquad
 p^1_j = 0 \mbox{ for all good }j \in [d] \qquad
      T= \frac{32 d\,L^2\gamma^2}{\eps^4 \sigma^2} \qquad
      \eta= \frac{\sqrt{2}\gamma}{L\sqrt{dT}}\qquad
      \]
    }
    \INDSTATE{For $t = 1, \ldots , T$:}
    \INDSTATE[2]{Observe the purchased bundle by the consumer $x^*(p^t)$}
    \INDSTATE[2]{Update price vector with projected subgradient descent:
      \[
      \hat{p}^{t+1}_j = p^t_j - \eta \left(\hat x_j - x^*(p^t)_j\right) \mbox{ for each }j\in [d],\qquad
 p^{t+1} = \proj{\cP} \left[\hat p^{t+1}\right]
      \]
}
    \INDSTATE{\textbf{Output:} $\hat p = 1/T\sum_{t=1}^T p^t.$}
    \end{algorithmic}
  \end{algorithm}

Our algorithm $\lp(\hat x, \eps)$ is given as Algorithm \ref{alg:learnprice}. We will prove:
\begin{theorem}
\label{thm:learnp}
Let $\hat x \in C$ be a target bundle and $\eps >
0$. Then $\lp(\hat x, \eps)$ outputs a price vector $\hat p$ such that
the induced bundle satisfies $\|\hat x - x^*(\hat p)\|\leq \eps$ and
the number of observations it needs is no more than 
\iffull\[\fi 
T= d\cdot \poly\left(\frac{1}{\eps}, \frac{1}{\sigma}, \gamma, \lambdav \right).
\iffull \]\fi
\end{theorem}


To analyze $\lp(\hat x, \eps)$, we will start by defining the following convex program whose solution is the target bundle $\hat x$.
\begin{align}
&\max_{x \in C} v(x) \qquad \label{eq:objective1}\\ \mbox{ such that }\qquad &x_j \leq \hat x_j \mbox{ for every good }j\in [d]\label{eq:demand}
\end{align}

Since $v$ is non-decreasing, it is not hard to see that $\hat x$ is
the optimal solution.  The \emph{partial Lagrangian} of this program is
defined as \iffull follows,
\[
\cL(x, p) = v(x) - \sum_{j = 1}^d p_j (x_j - \hat x_j),
\]
\fi
where $p_j$ is the dual variable for each
constraint~\eqref{eq:demand} and is interpreted as the price of good $j$. By strong
duality, we know that there is a $\emph{value}$ $\OPT$ such that
\begin{equation} \label{eq:minimaxvalue}
\max_{x \in C} \min_{p \in \RR^d_+} \cL(x, p) = \min_{p\in \RR^d_+}\max_{x \in C} \cL(x, p) = \OPT = v(\hat x).
\end{equation}
\iffull We know that $\OPT = v(\hat x)$ because $\hat x$ is the optimal solution to~\eqref{eq:objective1}-\eqref{eq:demand}. \fi

We can also define the \emph{Lagrange dual function} $g\colon \RR^d \to \RR$ to be
\iffull
\[
g(p) = \max_{x \in C} \cL(x, p).
\]
\fi We will show that an
approximately optimal price vector for $g$ approximately induces the
target bundle $\hat x$, and that $\lp(\hat x, \eps)$ is using
projected subgradient descent to find such a solution to $g$. In order
to reason about the convergence rate of the algorithm, we restrict the
space of the prices to the following bounded set:
\begin{equation}
\label{eq:bddprice}
\cP = \left\{p\in \RRP^d \mid \|p\|_2 \leq \sqrt{d}\mess\right\}.
\end{equation}

First, we can show that the minimax value of the Lagrangian remains
closed to $\OPT$ even if we restrict the prices to the set
$\cP$.\sw{changed}
\begin{lemma}
\label{lem:BPduality}
There exists a value $\RO$ such that
\[
\max_{x\in C} \min_{p\in \cP} \cL(x, p) = \min_{p\in \cP}
\max_{x\in C} \cL(x, p) = \RO.
\]
Moreover, $v(\hat x) \leq \RO \leq v(\hat x) +
\frac{\eps^2\sigma}{4}$.
\end{lemma}
\iffull
\begin{proof}
  Since $C$ and $\cP$ are both convex and $\cP$ is also compact,
  the minimax theorem~\cite{sion1958} shows that there is a value
  $\textrm{R-OPT}$ such that
  \begin{equation} \label{eq:restrictedvalue}
  \max_{x\in C} \min_{p\in \cP} \cL(x, p) = \min_{p\in \cP}
  \max_{x\in C} \cL(x, p) = \textrm{R-OPT}.
  \end{equation}
  Since $\cP \subseteq \RRP^{d}$, by~\eqref{eq:minimaxvalue}, we have
  $\textrm{R-OPT} \geq v(\hat x)$.  Thus, we only need to show that
  $\textrm{R-OPT} \leq v(\hat x) + \alpha$, where $\alpha =
  \eps^2\sigma/4$.
Let $(x\bl, p\bl)$ be a pair of minimax strategies
for~\eqref{eq:restrictedvalue}.  That is
\[
x\bl \in \argmax_{x \in C} \min_{p \in \cP} \cL(x, p) \qquad \textrm{and} \qquad
p\bl \in \argmin_{p \in \cP} \max_{x \in C} \cL(x, p)
\]
 It suffices to show
that $\cL(x\bl, p\bl)\leq v(\hat x) + \alpha$. Suppose not, then we have
\[
v(\hat x)  + \alpha < \cL(x\bl, p\bl) = \min_{p\in \cP} \cL(x\bl, p) = v(x\bl) - \max_{p\in \cP} \langle p,
x\bl -\hat x\rangle \leq v(x\bl).
\]

Now consider the bundle $y$ such that $y_j = \max\{x\bl_j, \hat x_j\}$
for each $j\in [d]$. It is clear that $v(y) \geq v(x\bl) > v(\hat
x)$. Let $L =\mess$, then we can construct the following price vector
$p' \in \cP$ such that $p'_j = L$ for each good $j$ with $x\bl_j >
\hat x_j$, and $p'_j = 0$ for all other goods.
Since we assume that $v$ is $(\lambdav, \beta)$-H\"{o}lder continuous
with respect to $\ell_2$ norm, we have
\[
v(x\bl) - v(\hat x) \leq v(y) - v(\hat x) \leq \lambdav\|y - \hat
x\|_2^\beta \leq \lambdav \|y - \hat x\|_1^\beta
\]
It follows that
\begin{align*}
v(\hat x) + \alpha < \cL(x\bl , p\bl) &\leq \cL(x\bl, p')\\
  &= v(x\bl) - \langle p', x\bl - \hat x\rangle\\
&= v(x\bl) - \sum_{j: x\bl_j > \hat x_j} L\, (y - \hat x_j)\\
&= v(x\bl) - L \|y -\hat x\|_1 
\leq v(y) - L \|y - \hat x\|_2
\end{align*}
Suppose that $\|y - \hat x\|_2 \geq 1$ or $\beta = 1$, we know that $\|y - \hat
x\|_2^\beta \leq \|y - \hat x\|_2$. This means $v(\hat x) + \alpha
< v(y) - L\|y - \hat x\|_2^\beta \leq v(y) - \lambdav \|y - \hat x\|_2^\beta\leq v(\hat x)$, a
contradiction.

Next suppose that $\|y - \hat x\|_2 < 1$ and $\beta\in (0,1)$. We also have that
\begin{align*}
  \alpha  &< v(y) - v(\hat x)  - L\|y - \hat x\|_2\\
  &\leq \lambdav \|y - \hat x\|_2^\beta - L \|y - \hat x\|_2\\
  &\leq \lambdav \, \|y-\hat x\|_2^\beta \left(1 - \frac{L}{\lambdav} \|y - \hat x\|_2^{1-\beta} \right)
\end{align*}
Since $\alpha > 0$, it must be that$\left(1 - \frac{L}{\lambdav} \|y -
\hat x\|_2^{1-\beta} \right)$ is also positive, and so $\|y - \hat
x\|_2 < \left(\frac{\lambdav}{L} \right)^{1/(1-\beta)}$. By the choice
of our $L$,
\begin{align*}
\alpha < \lambdav \left(\frac{\lambdav}{L} \right)^{\beta/(1-\beta)}  = \frac{\eps^2 \sigma}{4}= \alpha
\end{align*}
which is a contradiction. Therefore, the minimax value
of~\eqref{eq:restrictedvalue} is no more than $v(\hat x) + \alpha$.
\end{proof}
\else
\fi

The preceding lemma shows that $\hat x$ is a primal optimal solution
(even when prices are restricted).  Therefore, if $\hat p = \argmin_{p
  \in \cP} g(p)$ are the prices that minimize the Lagrangian dual, we
must have that $\hat x = x^*(\hat p)$ is the induced bundle at prices
$\hat p$.  The next lemma shows that if $p'$ are prices that
approximately minimize the Lagrangian dual, then the induced bundle
$x^*(p')$ is close to $\hat x$.

\begin{lemma}
\label{lem:approxP}
  Let $p'\in \cP$ be a price vector such that $g(p') \leq \min_{p\in
    \cP} g(p) + \alpha$.  Let $x' = x^*(p')$ be the induced bundle at
  prices $p'$.  Then $x'$ satisfies \iffull
  \[
  \|x' - \hat x\| \leq  2\sqrt{\alpha/\sigma}.
  \]
\fi
\end{lemma}

\iffull
\begin{proof}

Let $\RO$ denote the Lagrangian value when we restrict the price
space to $\cP$. From~\Cref{lem:BPduality}, we have that $\RO =
\min_{p\in \cP} g(p) \in [ v(\hat x), v(\hat x) + \alpha]$.  By
assumption, we also have 
$$g(p') = \cL(x', p') \leq \RO + \alpha \leq v(\hat x) +2\alpha.$$
Note that $\cL(\hat x, p') = v(\hat x) - \langle p', \hat x - \hat
x\rangle = v(\hat x)$ and $x'$ is the maximizer for $\cL(\cdot, p')$,
so it follows that
  \[
 0\leq   \cL(x', p')  -  \cL(\hat x , p') \leq 2 \alpha.
  \]
Since we know that $v$ is a $\sigma$-strongly concave function over
$C$, the utility function $u(\cdot, p') = v(\cdot) - \langle p', \cdot
\rangle$ is also $\sigma$-strongly concave over $C$.\footnote{If
  $f(\cdot)$ is a $\sigma$-strongly concave function over $C$ and
  $g(\cdot)$ is a concave function over $C$, then $(f+g)(\cdot)$ is a
  $\sigma$-strongly concave function over $C$.}  Then we have the
following by~\Cref{lem:sconvex} and the above argument,
  \begin{equation}\label{eq:closeness}
  2\alpha \geq \cL(x', p') - \cL(\hat x, p') = u(x', p') - u(\hat x, p') \geq \frac{\sigma}{2}\|x' - x\|^2
  \end{equation}
This means $\|x' - x\| \leq 2\sqrt{\alpha/\sigma}$.
\end{proof}
\fi Based on~\Cref{lem:approxP}, we can reduce the problem of finding
the appropriate prices to induce the target bundle to finding the
approximate optimal solution to $\argmin_{p\in \cP} g(p)$. Even though
the function $g$ is unknown to the producer (because $v$ is unknown),
we can still approximately optimize the function using projected
subgradient descent if we are provided access to subgradients of
$g$. The next lemma shows that the bundle $x^*(p)$ purchased by the
consumer gives a subgradient of the Lagrange dual objective function
at $p$.

\begin{lemma}
\label{lem:grad}
Let $p$ be any price vector, and $x^*(p)$ be the induced bundle. Then
\iffull
\[\left(\hat x - x^*(p)\right)\in \partial g(p).\]
\fi
\end{lemma}
\iffull
\begin{proof}
Given $x' = \arg\max_{x\in [0,1]^d} \cL(x, p)$, we know by the
envelope theorem that a subgradient of $g$ can be obtained as follows
\[
\frac{\partial g}{\partial p_j} = \hat x_j - x'_j\quad \mbox{ for each }j\in[d].
\]
Note that $x'$ corresponds to the induced bundle of $p$ because
\begin{align*}
  x' &= \argmax_{x \in C} \cL(x, p)\\
  &= \argmax_{x \in C} \left[v(x) - \langle p, x - \hat x\rangle \right]\\
  &= \argmax_{x \in C} \left[ v(x) - \langle p , x\rangle\right]\\
  &= \argmax_{x \in C} u(x, p) = x^*(p)
\end{align*}
Therefore, the vector $(\hat x - x^*(p))$ is a subgradient of $g$ at the price vector $p$.
\end{proof}
\fi
\iffull
Now that we know the subgradients of the function $g$ at $p$ can
be easily obtained from the induced bundle purchased by the consumer,
it remains to observe that Algorithm $\lp(\hat x, \eps)$ is performing
projected gradient descent on the Lagrange dual objective, and to
analyze its convergence.

\begin{proof}[Proof of Theorem \ref{thm:learnp}]
By~\Cref{lem:approxP}, it suffices to show that the price
vector $\hat p$ returned by projected gradient descent satisfies
\[
g(\hat p) \leq \min_{p\in \cP} g(p) + \frac{\eps^2\sigma}{4}.
\]
Note that the set $\cP$ is contained in the $\ell_2$ ball centered at
$\mathbf{0}$ with radius $L$. Also, for each $p^t$, the
subgradient we obtain is bounded: $\|\hat x -x^*(p^t)\|\leq
\sqrt{\|\hat x\|^2 + \|x^*(p^t)\|^2}\leq \sqrt{2}\gamma$ since
$\|C\|\leq \gamma$. Since we set
$$ T= \frac{32 d\, L^2\gamma^2}{\eps^4 \sigma^2} \qquad \eta=
\frac{\sqrt{2}\gamma}{L\sqrt{dT}}$$
we can apply the guarantee of projected
gradient descent from~\Cref{thm:gd}, which gives:
\[
g(\hat p) - \min_{p\in \cP} g(p) \leq \frac{\sqrt{2}L \gamma}{\sqrt{T}}
= \frac{\eps^2\sigma}{4}
\]
By~\Cref{lem:approxP}, we know that the resulting bundle $x^*(\hat p)$
satisfies that $\|\hat x - x^*(p)\| \leq \eps$.
\end{proof}
\fi

\iffull
\begin{remark}
Since noise tolerance is not required in this setting, it is possible
approximately induce the target bundle only using
\emph{poly-logarithmically} in $(1/\eps)$ number of observations. We
will give an ellipsoid-based variant of $\lp$ in
\iffull~\Cref{sec:ellip} \fi that achieves this guarantee.
\end{remark}
\fi

\subsection{Profit Maximization}
\label{sec:prof}


Finally, we will show how to combine the algorithm $\lp$ with the zeroth
order optimization algorithm $\zero$ to find the approximate
profit-maximizing price vector. At a high level, we will use $\zero$
to (approximately) optimize the profit function $r$ over the bundle space and use
$\lp$ to (approximately) induce the optimal bundle.

Before we show how to use $\zero$, we will verify that if we run the
algorithm $\lp$ to obtain prices $\hat p$ that approximately induce
the desired bundle $x$, and observe the revenue generated from prices
$\hat p$, we will indeed obtain an approximation to the revenue
function $r(x)$.

Recall from~\Cref{lem:best-price} that the profit function can be
written as a function of the bundle
\iffull
\[
r(x) = \langle \nabla v(x) , x \rangle - c(x)
\]
\fi as long as
the producer uses the profit maximizing price vector $\nabla v(x)$ to
induce the bundle $x$. However, the price vector returned by $\lp$
might not be the optimal price vector for the induced bundle.  In
order to have an estimate of the optimal profit for each bundle, we
need to guarantee that prices returned by $\lp$ are the profit
maximizing ones. To do that, we will restrict the bundle space that
$\zero$ is optimizing over to be the interior of $C$.  
Now we show that for every bundle in the interior of $C$, there is a unique price vector that induces that bundle.  Thus, these prices are the profit-maximizing prices inducing that bundle.
\begin{lemma}
\label{lem:intP}
Let $x'$ be a bundle in $\mathrm{Int}_{C}$. Then $\nabla v(x')$ is the unique price vector that induces $x'$.
\end{lemma}
\iffull
\begin{proof}
Let $p'$ be a price vector such that $x^*(p') = x$. Since
$\mathrm{Int}_C\subseteq C$, we must have
\[
x' = \argmax_{x\in \mathrm{Int}_C} \left[ v(x) - \langle p', x \rangle\right].
\]
By the definition of $\mathrm{Int}_C$, we know that there exists some
$\delta > 0$ such that the ball $\delta B_{x'}$ is contained in
$C$. Now consider the function $f\colon \RR^d \to \RR$ such that $f(x)
= u(x, p')$. It follows that $x'$ is a local optimum of $f$
neighborhood $\delta B_{x'}$. Since $f$ is continuously
differentiable, we must have $\nabla f(x') = \mathbf{0}$ by
first-order conditions. Therefore, we must have
\[
\nabla f(x') = \nabla v(x') - p' = \mathbf{0},
\]
which implies that $p' = \nabla v(x')$.
\end{proof}
\fi
Instead of using the interior itself, we will use a simple and efficiently computable proxy for the interior obtained by slightly shifting and contracting $C$.


\begin{claim}
\label{clm:deltaInt}
For any $0 <\delta <1/2$, let the set \iffull
 $$C_{\delta} = (1- 2 \delta) C + \delta \mathbf{1},$$
\fi
 where $\mathbf{1}$ denotes the
$d$-dimensional vector with 1 in each coordinate. Given~\Cref{ass:feasible}, $C_{\delta}$ is
contained in the $(\delta/2)$-interior of $C$.  \iffull That is, $C_{\delta}
\subseteq \mathrm{Int}_{C, \delta/2}.$\fi
\end{claim}

\iffull
\begin{proof}
Our goal is to show that $C_{\delta} + \delta B_{0} \subseteq C$,
where $B_0$ denote the unit ball centered at $\mathbf{0}$.  Any point
in $C_\delta + (\delta/2) B_{0}$ can be written as $x' + (\delta/2)\,
y'$ for $x' \in C_{\delta}$ and $y' \in B_{0}$.  We will show that $x'
+ (\delta/2)\, y' \in C$.  Since $x' \in C_{\delta}$, there exists $x \in
C$ such that
$$
x' = (1- 2\delta) x + \delta \mathbf{1}.
$$
Since $y' \in B_{0}$, there exists $y \in (0,1]^d$ such that
$$
\frac{1}{2}\, y' = 2y - \mathbf{1}.
$$ To see this, note that $(0,1]^d$ contains a ball of radius $1/4$
  whose center is $(1/2)\cdot \mathbf{1}$.  By
  Assumption~\ref{ass:feasible}, $C$ contains $(0,1]^d$, so $y \in C$.
  Therefore for some $x, y \in C$,
\begin{align*}
x' + (\delta/2)\, y'
&={} (1-2\delta) x + \delta \mathbf{1} + 2 \delta y - \delta \mathbf{1} \\
&={} \underbrace{(1-2 \delta) x + 2\delta y}_{\in C},
\end{align*}
where we used convexity of $C$.  Hence, $x' + (\delta/2)\, y' \in C$, as desired.
\end{proof}
\fi

We will let $\zero$ operate on the set $C_\delta$ instead of $C$, and
we first want to show that there is little loss in profit if we
restrict the induced bundle to $C_{\delta}$. \iffull The following is just a formal, quantitative version of of Assumption \ref{ass:cost}:
\begin{assumption}[Quantitative version of Assumption~\ref{ass:cost}]
The producer's cost function $c\colon \RRP^d \to \RR$ is
$\lambdac$-Lipschitz over the domain $C$ with respect to the $\ell_2$ norm:
for $x, x'\in C$,\iffull
\[
|c(x) - c(x')| \leq \lambdac \|x - x'\|.
\]
\else
$|c(x) - c(x')| \leq \lambdac \|x - x'\|.$
\fi
\end{assumption}\fi
\iffull

Given this assumption, the profit function is also H\"{o}lder continuous.

\begin{lemma}
  \label{lem:lipschitzbound}
For any $x, y\in C$ such that $\|x - y\|\leq 1$, the following holds
\[
|r(x) - r(y)| \leq (\lambdav + \lambdac) \|x - y\|^\beta.
\]
\end{lemma}

\begin{proof}
Recall the revenue component of the profit function is $\langle \nabla
v(x) , x \rangle$. Since $v$ is a concave and homogeneous function, we
know that the homogeneity degree satisfies $k\leq
1$. (See~\Cref{sec:homoproof} for a proof).  By Euler's theorem
(\Cref{thm:euler}),
\begin{equation}
\langle \nabla v(x) , x\rangle = k\cdot v(x).
\label{eq:euler}
\end{equation}
Since $v$ is $(\lambdav, \beta)$-H\"{o}lder continuous $C$, by
Equation \ref{eq:euler} we know that the revenue $\langle \nabla v(x)
, x\rangle$ is also $\lambdav$-H\"{o}lder continuous over
$C$. Furthermore, since the cost function $c$ is $\lambdac$-Lipschitz
over $C$, the profit function satisfies the following: for any $x, y
\in C$ such that $\|x - y\|\leq 1$, we have \[
|r(x) - r(y)| \leq  |\langle \nabla v(x), x\rangle - \langle \nabla v(y), y\rangle| + |c(x) - c(y)|
\leq \lambdav\|x - y\|^\beta + \lambdac\|x - y\|
\]
Since $\|x - y\|\leq 1$, we know that $\|x-y\|^\beta \geq \|x- y\|$,
so $|r(x) - r(y)| \leq (\lambdav + \lambdac) \|x - y\|^\beta$.
\end{proof}

We can bound the difference between the optimal profits in
$C_{\delta}$ and $C$.  \fi

\begin{lemma}
\label{lem:intApprox}
For any $0< \delta \leq 1/3\gamma$,
\iffull
$$
\max_{x\in C} r(x) - \max_{x\in C_\delta} r(x) \leq (3\delta\gamma)^\beta(\lambdav + \lambdac).
$$
\fi
\end{lemma}
\iffull
\begin{proof}
Let $x^*\in \arg\max_{x\in C} r(x)$. We know that $(1 - 2\delta) x^* +
\delta \mathbf{1}\in C_\delta$, and
\[
\| x^* - (1 - 2\delta) x^* - \delta\mathbf{1} \| \leq \delta \|2x^* -
\mathbf{1}\| \leq 3\delta \gamma.
\]
By~\Cref{lem:lipschitzbound}, we then have
\[
r(x^*) - r((1 - \delta) x^* + \delta\mathbf{1}) \leq
(3\delta\gamma)^\beta (\lambdav + \lambdac).
\]
Furthermore, we also know $\max_{x\in C_\delta} r(x) \geq r((1-
\delta)x^* + \delta \mathbf{1})$, so we have shown the bound
above.
\end{proof}
\fi

Now we focus on how to optimize the profit function $r$ over the set
$C_\delta$. Recall the algorithm $\zero$ requires approximate
evaluations for the profit function $r$. Such evaluations can be
implemented using our algorithm $\lp$: for each bundle $x\in
C_\delta$, run $\lp(x, \eps)$ to obtain a price vector $p$ such that
$\|x - x^*(p)\|\leq \eps$, and then the resulting profit $r(x^*(p))$
serves as an approximate evaluation for $r(x)$:
\iffull
\[
|r(x) - r(x^*(p))| \leq ( \lambdav + \lambdac) \eps^\beta.
\]
\fi

\iffull
\begin{algorithm}[h]
  \caption{Learning the price vector to optimize profit: $\op(C,
    \alpha)$}
 \label{alg:optprice}
  \begin{algorithmic}

    \STATE{\textbf{Input:} Feasible bundle space $C$, and target accuracy $\alpha$}

    \INDSTATE{Initialize:
      \[
      \eps= \min\left\{\left(\frac{\alpha}{\lambda(d + 1 +  (12\gamma)^\beta)}\right)^{1/\beta}, \frac{1}{12\gamma}\right\}\qquad
      \delta = 4\eps \qquad
      \alpha' =  d\eps^\beta( \lambdav + \lambdac)
      \]

restricted bundle space $C_\delta = (1 -
      2\delta) C + \delta\mathbf{1}$ and number of iterations $T=
      \tilde O(d^{4.5})$ }
    \INDSTATE{For $t = 1, \ldots , T$:}
    \INDSTATE[2]{$\zero(\alpha', C_\delta)$ queries the profit for bundle $x^t$}
    \INDSTATE[2]{Let $ p^t = \lp(x^t, \eps)$ and observe the induced bundle $x^*( p^t)$}
    \INDSTATE[2]{Send $r(x^*(p^t))$ to $\zero(\alpha', C_\delta)$ as an approximate evaluation of $r(x^t$)}
    \INDSTATE{$\hat x = \zero(\alpha', C_\delta)$}
    \INDSTATE{$\hat p = \lp(\hat x, \eps)$}
    \INDSTATE{\textbf{Output:} the last price vector $\hat p$}
    \end{algorithmic}
  \end{algorithm}
\fi

\begin{theorem}
Let $\alpha > 0$ be the target accuracy. The instantiation $\op(C,
\alpha)$ computes a price vector $\hat p$ such that the expected
profit \iffull \[ \fi \ex{r(\hat p)} \geq \max_{p\in \RRP^d}
r(p) - \alpha, \iffull \] \fi the number of times it calls the
algorithm $\lp$ is bounded by $\tilde O(d^{4.5})$, and the total
observations it requires from the consumer is $\poly(d,
1/\alpha)$.~\footnote{ In~\Cref{sec:ellip}, we give a variant of the
  algorithm with query complexity scaling poly-logarithmically in
  $1/\alpha$.}
\label{thm:optProf}
\end{theorem}
\iffull
\begin{proof}
First we show that each induced bundle $x^*( p^t)$ is in the interior
$\mathrm{Int}_C$. Note that in the algorithm, we have $\eps= 
\delta/4$. By the guarantee of $\lp$ in~\Cref{thm:learnp}, we have
that
\[
\|x^t - x^*(p^t)\| \leq \eps =\delta/4.
\]
By~\Cref{clm:deltaInt}, we know that $x^t\in \mathrm{Int}_{C,
  \delta/2}$, so the ball of radius $\eps$ centered at $x^t$ is
contained in $C$, and hence $x^*(p^t)$ is in the interior of
$C$. By~\Cref{lem:intP} and~\Cref{lem:best-price}, each vector $p^t =
\nabla v(x^*(p^t))$ is the profit-maximizing prices for the induced
bundle $x^*(p^t)$, so the profit the algorithm observes is indeed
$r(x^*(p^t))$.

Next, to establish the accuracy guarantee, we need to bound two
sources of error. First, we need to bound the error from $\zero$. To
simplify notation, let $\lambda = ( \lambdav + \lambdac)$. Recall
from~\Cref{lem:lipschitzbound} that the approximate profit evaluation
$r(x^*(p^t))$ satisfies
\[
|r(x^t) - r(x^*(p^t))| \leq \lambda \eps^\beta.
\]
By the accuracy guarantee in~\Cref{lem:zero}, the final queried bundle
$\hat x$ satisfies
\[
\Expectation[ r(\hat x)] \geq \max_{x\in C_\delta} r(x) - d\lambda\eps^\beta.
\]
Since we know that $|r(\hat x) - r(x^*(\hat p))| \leq \lambda \eps$, we also have
\[
\Expectation [r(x^*(\hat p))] \geq \max_{x\in C_\delta} r(x) - (d + 1)\lambda\eps^\beta.
\]
Next, as we are restricting the bundle space to $C_\delta$, there
might be further loss of profit. Note that $\delta = 4\eps \leq
1/3\gamma$, so we can bound it with~\Cref{lem:intApprox}:
\[
\Expectation[ r\left(x^*(\hat p)\right)] \geq \max_{x\in C} r(x) -
\lambda\left[(d + 1)\eps^\beta + (3\delta\gamma)^\beta\right] = \max_{x\in C} r(x) -
\lambda\left[(d + 1)\eps^\beta + (12\eps\gamma)^\beta\right].
\]
If we plug in our setting for parameter $\eps$, we recover the desired
bound since $r(x^*(\hat p)) = r(\hat p)$ and $\max_{x\in C} r(x) =
\max_{p\in \RRP^d} r(p)$.

Finally, we need to bound the total number of observations the
algorithm needs from the consumer. In each iteration, the
instantiation $\lp(x^t, \eps)$ requires number of observations bounded
by according to~\Cref{thm:learnp}
\[
T' = d\, \cdot \poly\left(\frac{1}{\eps}, \frac{1}{\sigma}, \gamma, \lambdav \right)
\]
Therefore, after plugging in $\eps$, we have that the total number of observations $\op$
needs is bounded by
\[
O(T' \times T) = \poly(d, 1/\alpha)
\]
(hiding constants $\lambdac, \lambdav, \sigma, \gamma$).
\end{proof}
\else
\fi

\iffull

\section{General Framework of Stackelberg Games}
\label{sec:general}
Now that we have worked out a concrete application of our method in
the context of learning to maximize revenue from revealed preferences,
we will abstract our techniques and show how they can be used to solve
a general family of Stackelberg games in which the objective of the
follower is unknown to the leader.  Along the way, we will also
generalize our technique to operate in a setting in which the follower
responds to the leaders actions by only \emph{approximately}
maximizing her utility function. In addition to generalizing the
settings in which our approach applies, this avoids a technical
concern that might otherwise arise -- that bundles maximizing strongly
concave utility functions might be non-rational. In addition to being
able to handle approximations to optimal bundles that would be induced
by taking a rational approximation, we show our method is robust to
much larger errors.

\sw{Will be Noisy}

In our general framework, we consider a \emph{Stackelberg game} that
consists of a \emph{leader} with action set $\cA_L$ and a
\emph{follower} with action set $\cA_F$. Each player has a utility
function $U_L, U_F\colon \cA_L\times\cA_F \to \RR$.  In the
corresponding Stackelberg game, the leader chooses an action $p \in
\cA_L$, and then the follower chooses a $\zeta$-best response $x'(p)$
such that \sw{changed here}
\[ U_F(p, x'(p)) \geq  U_F(p, x^*(p)) - \zeta,\]
where $x^*(p) = \argmax_{x\in \cA_F} U_F(p, x)$ is the follower's
\emph{exact} best-response. Note that when $\zeta = 0$, $x'(p) =
x^*(p)$.

The example of maximizing revenue from revealed preferences is a
special case of this framework.  The producer is the leader and his
action space consists of prices $p$ and the follower is the consumer
and her action space is the bundle $x$ she purchases.  The producer's
utility for a pair $(p, x)$ is his revenue minus the cost of producing
$x$ and the consumer's utility is her value for $x$ minus the price
she pays.

In general, we consider solving the leader's optimization
problem---find $p\in \cA_L$ such that $U_L(p, x^*(p))$ is
(approximately) maximized.\sw{still debating whether it should be
  $U_L(p, x^*(p))$} Formally, we consider a sub-class of Stackelberg
games that have the following structure.

\begin{definition}
\label{def:problem}
An instance is a \emph{Stackelberg game} $\cS(\cA_L, \cA_F, \phi)$ which
consists of two players---the \emph{leader} and the \emph{follower} such that:
\begin{itemize}
\item the leader has action set $\cA_L\subseteq \RR^d$,
the follower has action set $\cA_F\subseteq \RR^d$, both of which are
convex and compact;
\item the follower's utility function $U_F\colon \cA_L\times \cA_F\to \RR$
 takes the form
\[
U_F(p, x) = \phi(x) - \langle p, x\rangle,
\]
where $\phi\colon \RR^d \to \RR$ is a strongly concave, differentiable
function unknown to the leader;
\item the leader's utility function $U_L\colon \cA_L \times \cA_F\to \RR$ is
an unknown function.
\end{itemize}
The optimization problem associated with the game instance is
$\max_{p\in \cA_L} \psi(p, x^*(p))$.
\end{definition}

Our first step to solve the problem is to rewrite the leader's utility
function so that it can be expressed as a function only in the
follower's action. For each action of the follower $x\in \cA_F$, the
set of leader's actions that induce $x$ is
\[
P^*(x) = \{p \in \cA_L\mid x^*(p) = x\}.
\]
Among all of the leader's actions that induce $x$, the optimal one is:
\[
p^*(x) = \argmax_{p\in P^*(x) } U_L(p, x),
\]
where ties are broken arbitrarily.
We can then rewrite the leader's objective as a function of only
$x$:
\begin{equation} \label{eq:psi}
\psi(x) = U_L(p^*(x) , x).
\end{equation}

Note that to approximately solve the leader's optimization problem, it
is sufficient to find the follower's action $\hat x\in \cA_F$ which
approximately optimizes $\psi_F(\cdot)$, together with the action
$\hat p\in \cA_L$ that approximately induces $\hat x$. Before we
present the algorithm, we state the assumptions on the utility
functions of the two players that we will need.
\begin{assumption}
\label{ass:general}
The game $\cS(\cA_L, \cA_F, \phi)$ satisfies the following properties.
\begin{enumerate}
\item The function $\psi\colon \cA_L \to \RR$ defined in~\eqref{eq:psi} is concave and $\lambda_L$-Lipschitz;
\item The function $\phi\colon \cA_F \to \RR$ is non-decreasing, $\sigma$-strongly concave and $\lambda_F$-Lipschitz;
\item The action space of the leader $\cA_L$ contains the following set
\begin{equation}
\cP = \{p\in \RRP^d \mid \|p\| \leq \sqrt{d} \lambda_F\};
\label{eq:P}
\end{equation}
\item The action space of the follower $\cA_F$ has bounded diameter, $\|\cA_F\| \leq \gamma$.
\end{enumerate}

\end{assumption}

\subsection{Inducing a Target Action of the Follower}
\label{sec:generalinduce}
We first consider the following sub-problem.  Given a target action
$\hat x$ of the follower we want to learn an action $\hat p$ for the leader such that the
induced action satisfies
\[
\| x'(\hat p) - \hat x\| \leq \eps.
\]
We now give an algorithm to learn $\hat p$ that requires only
polynomially many observations of the follower's $\zeta$-approximate
best responses.

\begin{algorithm}[h]
  \caption{Learning the leader's action to induce a target follower's action: $\LL(\hat x, \eps)$}
 \label{alg:learnlead}
  \begin{algorithmic}
    \STATE{\textbf{Input:} A target follower action $\hat x\in \cA_F$, and target accuracy $\eps$}
    \INDSTATE{Initialize:
 restricted action space $\cP=\{p\in  \RR_+^{d} \mid \|p\|\leq \sqrt{d}\lambda_F\}$
      \[
      p^1_j = 0 \mbox{ for all }j \in [d] \qquad
      T = \left(\frac{16\sqrt{2d}\lambda_F \gamma}{\eps^2\sigma - 4\zeta}\right)^2 \qquad
\eta= \frac{\sqrt{2}\gamma}{\sqrt{d}\lambda_F\sqrt{T}}
      \]
    }
    \INDSTATE{For $t = 1, \ldots , T$:}
    \INDSTATE[2]{Observe the induced action by the follower $x^*(p^t)$}
    \INDSTATE[2]{Update leader's action:
      \[
      \tilde{p}^{t+1}_j = p^t_j - \eta \left(\hat x_j - x^*(p^t)_j\right) \mbox{ for each }j\in [d],\qquad
p^{t+1} = \proj{\cP} \left[\hat p^{t+1}\right]
      \]
}
    \INDSTATE{\textbf{Output:} $\hat p = 1/T\sum_{t=1}^T p^t.$}
    \end{algorithmic}
  \end{algorithm}

\begin{theorem}
\label{thm:generalmain}
Let $\hat x \in \cA_F$ be a target follower action and $\eps >
0$. Then $\LL(\hat x, \eps)$ outputs a leader action $\hat p$ such
that the induced follower action satisfies $\|\hat x - x'(\hat
p)\|\leq \eps$ and the number of observations it needs is no more
than \[T=O \left(\frac{d \lambda_F^2\gamma^2}{\eps^4
  \sigma^2}\right)\] as long as $\eps > 2\sqrt{2\zeta/\sigma}$.
\end{theorem}

\subsection{Optimizing Leader's Utility}
Now that we know how to approximately induce any action of the
follower using $\LL$, we are ready to give an algorithm to optimize
the leader's utility function $U_L$. Recall that we can write the
$U_L$ as a function $\psi$ that depends only of the follower's
action. In order to obtain the approximately optimal utility
$\psi(x)$, the leader must play the \emph{optimal} action $p$ that
induces the follower to play approximately $x$.

\begin{assumption}
\label{ass:consistent}
For any $\hat x\in \cA_F$ and $\eps > 0$, the instantiation $\LL(\hat
x, \eps)$ returns $\hat p$ such that
\[
\hat p = p^*\left(x^*(\hat p)\right).
\]
\end{assumption}
Whenever this assumption holds, we can use $\LL$ to allow the leader to obtain utility $U_L(\hat p , x^*(\hat
p)) = \psi(x^*(\hat p))$.

While~\Cref{ass:consistent} appears to be quite strong, we can often achieve it. Recall that we were able to satisfy~\Cref{ass:consistent} in our revealed preferences application by operating in the interior of the feasible region of the follower's action space, and we can similarly do this in our principal-agent example. Moreover, it is trivially satisfied whenever the leader's objective function \emph{depends only on the follower's action}, since in this case, every leader-action $p$ which induces a particular follower-action $x$ is optimal. This is the case, for example, in our routing games application in~\Cref{sec:flow}.

Now we will show how to use the algorithm $\zero$ to find an
approximate optimal point for the function $\psi$. First, we will use
$\LL$ to provide approximate function evaluation for $\psi$ at each
$\hat x\in \cA_F$: our algorithm first runs $\LL(\hat x, \eps)$ to
learn a price vector $\hat p$, and we will use the observed function
value on the induced follower's approximate best response
$\psi(x'(\hat p))$ as an approximation for $\psi(\hat x)$. Since
$\LL$ guarantees that $\|x'(\hat p) - \hat x\| \leq \eps$, by the
Lipschitz property of $\psi$ we have
\[
|\psi(\hat x) - \psi(x'(\hat p))| \leq \lambda_L \eps.
\]
With these approximate evaluations, $\zero$ can then find a
$(d\lambda_L \eps)$-approximate optimizer of
$\psi$ with only $\tilde O(d^{4.5})$ iterations by~\Cref{lem:zero}.
The full algorithm is presented in~\Cref{alg:lopt}.

\begin{algorithm}[h]
  \caption{Leader learn to optimize: $\lop(\cA_F, \alpha)$}
 \label{alg:lopt}
  \begin{algorithmic}

    \STATE{\textbf{Input:} Follower action space $C$, and target
    accuracy $\alpha$} \INDSTATE{Initialize: number of iterations
    $T= \tilde O(n^{4.5})$ and $\eps = \frac{\alpha}{ \lambda_L (d +
    1)}$}
    \INDSTATE{For $t = 1, \ldots , T$:}
    \INDSTATE[2]{$\zero(d\eps\lambda_L, \cA_F)$ queries the objective value for action $x^t\in \cA_F$}
    \INDSTATE[2]{Let $ p^t = \LL(x^t, \eps)$ and observe the induced action $x'( p^t)$}
    \INDSTATE[2]{Send $\psi(x'(p^t))$ to $\zero(d\eps\lambda_L, C_\delta)$ as an approximate evaluation of $\psi(x^t$)}
    \INDSTATE{$\hat x = \zero(d\eps\lambda_L, \cA_F)$}
    \INDSTATE{$\hat p = \LL(\hat x, \eps)$}
    \INDSTATE{\textbf{Output:} the leader action $\hat p$}
    \end{algorithmic}
  \end{algorithm}

\begin{theorem}
\label{thm:genmain}
Let $\alpha > 0$ be the target accuracy. The instantiation
$\lop(\cA_F, \alpha)$ computes a leader action $\hat p$ along with its
induced follower action $x^*(\hat p)$ that satisfies
\[\Expectation[
U_L(\hat p, x^*(\hat p))] \geq \max_{p\in \cA_L} U_L(p, x^*(p)) - \alpha,
\] and the number of observations the algorithm requires of the follower is bounded by
\[
\tilde O\left(  \frac{d^{9.5}}{\alpha^4}\right),
\]
as long as $\alpha \geq \Omega(d\lambda_L \sqrt{{\zeta}/{\sigma}})$.
\end{theorem}

\section{Optimal Traffic Routing from Revealed Behavior}
\label{sec:flow}
In this section, we give the second main application of our technique
discussed in the introduction: how to find tolls to induce an
approximately optimal flow in a non-atomic traffic routing game when
the latency functions are unknown.

A nonatomic routing game $\cG(G, \ell, \cD)$ is defined by a graph
$G=(V,E)$, latency function $\ell_e$ on each edge $e\in E$, and the
source, destination and demands for $n$ commodities: $\cD = \{(s_i,
t_i, k_i)\}_{i\in [n]}$. The latency function $\ell_e\colon \RR_+
\rightarrow [0,1]$ represents the delay on each edge $e$ as a function
of the total flow on that edge. For simplicity, we assume $\sum_{i =
  1}^n k_i = 1$, and we let $m$ denote the number of edges $|E|$.

For each commodity $i$, the demand $k_i$ specifies the volume of flow
from $s_i$ to $t_i$ routed by (self-interested) agents. The game is
nonatomic: infinitely many agents each control only an infinitesimal
amount of flow and each agent of type $i$ selects an action (an
$s_i$-$t_i$ path) so as to minimize her total latency. The aggregate
decisions of the agents induce a multicommodity flow $(f^i)_{i\in
  [n]}$, with each vector $f^i =(f^i_e)_{e\in E}\in \cF_i$, where
$\cF_i$ is the flow polytope for the $i$'th commodity:
\[
\cF_i = \left\{f^i\in \RR_+^{m}\mid \sum_{(v,w)\in E} f^i_{vw} = \sum_{(u, v)\in E}
f^i_{uv}, \forall v\in V\setminus\{s_i, t_i\}, \quad \sum_{(s_i,
  w)\in E} f^i_{s_iw} - \sum_{(u,s_i)\in E} f^i_{u, s_i} = k_i\right\}
\]
Let $\cF = \{f = \sum_{i=1}^n f^i\mid f^i\in \cF_i \mbox{ for each }
i\}$ denote the set of feasible flows. A flow $f$ defines a latency
$\ell_e(f_e)$ on each edge $e$. Given a path $P$, we write $\ell_P(f)
= \sum_{e \in P} \ell_e(f_e)$ to denote the sum latency on all edges
in the path.
\iffull
A \emph{Nash} or \emph{Wardrop} equilibrium is defined as follows:\fi
\begin{definition}[Wardrop equilibrium]\label{nash-flow}
A multicommodity flow $\hat f$ is a Wardrop equilibrium of a routing
game if it is feasible and for every commodity $i$, and for all
$s_i$-$t_i$ paths $P, Q$ with $\hat f^i_P > 0$, we have $\ell_P(\hat
f) \leq \ell_Q(\hat f)$.
\end{definition}

Crucial to our application is the following well known lemma, which
states that a Wardrop equilibrium can be found as the solution to a
optimization problem (convex whenever the latencies are
non-decreasing), which minimizes a potential function associated with
the routing game

\begin{lemma}[\cite{MS96}]
A Wardrop equilibrium can be computed by
solving the following optimization problem:
\begin{align*}
\min_{f\in \cF} \, \Phi(f) &\coloneqq \sum_e \int_0^{f_e} \ell_e(x)\, dx
\end{align*}
Whenever the latency functions $\ell_e$ are each non-decreasing, this
is a convex program.  We call $\Phi$ the potential function of the
routing game.
\end{lemma}

Now suppose there is a municipal authority which administers the network
and wishes to minimize the social cost of the equilibrium flow:
\[
\Psi(f) = \sum_{e\in E} f_e \cdot \ell_e(f_e).
\]
The authority has the power to impose constant tolls on the edges.
A toll vector $\tau = (\tau_e)_{e\in E}\in \RR_+^{m}$ induces a new latency function on each edge: $\ell_e^\tau(f_e) =
\ell(f_e) + \tau_e$, which gives rise to a different routing game
$\cG(G, \ell^\tau, \cD)$ with a new potential function $\Phi^\tau$. In
particular, the equilibrium flow $f^*(\tau)$ induced by the toll
vector is the Wardrop equilibrium of the tolled routing game:
\[
f^*(\tau) = \argmin_{f\in \cF} \Phi^\tau(f) = \argmin_{f\in \cF}
\left[ \sum_{e\in E} \int_0^{f_e} (\ell_e(x) + \tau_e) dx \right] =  \argmin_{f\in \cF}
\left[\Phi(f) + \sum_{e\in E} \tau_e \cdot f_e\right].
\]
While the latency functions are unknown to the authority, his goal is
to find a toll vector $\hat\tau$ such that the induced flow $f^*(\hat
\tau)$ approximately minimizes the total congestion function $\Psi$.

We can formulate this problem as an instance of the type of Stackelberg game we defined
in~\Cref{def:problem}, where the authority is the leader, and there is a single ``flow'' player minimizing the game's potential function, serving the role of the follower. We will refer to them as the
toll player and the flow player respectively. In our setting:
\begin{enumerate}
\item The toll player has action set $\tau\in \RRP^m$ and the flow player has action set $\cF$;
\item The flow player has a utility function $U_F\colon \RRP^m \times \cF\to \RR$ of the form
  \[
  U_F(\tau, f) = -\Phi(f) - \langle \tau , f \rangle;
  \]
\item The toll player has a utility function $U_L\colon \RRP^m\times \cF \to \RR$ of the form
  \[
  U_L(\tau, f) = -\Psi(f).
  \]
\end{enumerate}
Now we will apply the tools in~\Cref{sec:general} to solve this
problem. Before we begin, we will impose the following assumptions on
the latency functions to match with~\Cref{ass:general}. We need two types of assumptions: one set to let us find tolls to induce a target flow, and another to guarantee that once we can induce such flows (and hence implement a ``flow cost oracle''), we can optimize over flows.

To find tolls to induce a target flow, we require that the potential function $\Phi$ be strongly convex in the flow variables. The following conditions are sufficient to guarantee this:
\begin{assumption}
\label{ass:potential-convex}
For each edge $e \in E$, $\ell_e$ is differentiable and has derivative bounded away from zero: there exists some $\sigma > 0$ such that for all $x \in [0,1]$, $\ell_e'(x) \geq \sigma$.
\end{assumption}
Recall that the potential function $\Phi(x)$ is a function on $m$
variables $(f_e)_{e\in E}$, and it's Hessian $\nabla^2 \Phi$ at each
$f\in \cF$ is a diagonal matrix with entries $\ell_e'(f_e) \geq
\sigma$. Therefore, we know that $\nabla^2 \Phi(f) \succeq \sigma I$
for any $f\in \cF$, and so under Assumption
\ref{ass:potential-convex}, $\Phi$ is a $\sigma$-strongly convex
function over $\cF$. Note that the only condition we really require is
that the potential function be strongly convex, and there are weaker
conditions that imply this, but we state Assumption
\ref{ass:potential-convex} because of its simplicity.

Once we can implement a flow oracle, we need to be able to use a
bandit convex optimization algorithm to optimize social cost over
flows. Hence, we require that the social cost function be convex and
Lipschitz. The following assumptions are sufficient to guarantee this:
\begin{assumption}
  For each edge $e \in E$, $\ell_e$ is convex and $(\lambda/m)$-Lipschitz continuous over $[0,1]$.
\end{assumption}
Note that this guarantees that $\Psi$ is $\lambda$-Lipschitz over $\cF$.

We first show that we can use the algorithm $\LL$ to learn a toll
vector to induce any flow as a Wardrop equilibrium.

\begin{lemma}
\label{lem:bhaskar-compare} Fix any non-atomic routing game satisfying Assumption \ref{ass:potential-convex}.
  Let $\hat f\in \cF$ in a target flow and $\eps > 0$. Then the
  instantiation $\LL(\hat f, \eps)$ outputs a toll vector $\hat \tau$
  such that the induced Wardrop equilibrium flow $f^*(\hat \tau)$
  satisfies $\|\hat f - f^*(\hat \tau)\| \leq \eps$, and the number of
  observations on the flow behavior it needs is no more than
  \[
   O\left(\frac{m^3}{\eps^4\sigma^2}\right).
  \]
\label{lem:targetflow}
\end{lemma}

\begin{proof}
  Before we apply~\Cref{thm:generalmain}, we still need to show that
  the potential function $\Phi$ of the original routing game (without
  tolls) is Lipschitz over $\cF$. Note that this does not require any assumptions on the latency functions $\ell_e$ other than that they are bounded in $[0,1]$. Let $f, g\in \cF$, then we can write
  \begin{align*}
\left    |\Phi(f) - \Phi(g) \right| &= \left|\sum_e \left( \int_0^{f_e} \ell_e(x)\,dx -
  \int_0^{g_e} \ell_e(x)\, dx \right)\right| \\
 & = \left|\sum_{e\in E} \int_{g_e}^{f_e} \ell_e(x)\, dx\right|\\
 &\leq \sum_{e\in E} \max\{\ell_e(f_e), \ell_e(g_e)\}\left|f_e - g_e \right|\\
  &\leq \sum_e |f_e - g_e| \leq \sqrt{m} \|f - g\|,
\end{align*}
where the last inequality follows from the fact that $\|x\|_1 \leq
\sqrt{m} \|x\|_2$ for any $x\in \RR^m$. Also, observe that each flow
vector in $\cF$ has norm bounded by $\sqrt{m}$. Therefore, we know
that $\Phi$ is a $\sqrt{m}$-Lipschitz function. Then we can
instantiate~\Cref{thm:generalmain} and obtain the result above.
\end{proof}


Now we can instantiate~\Cref{thm:genmain} and show that $\lop$ can
find a toll vector that induces the approximately optimal
flow.

\paragraph{Pre-processing Step}{
The set $\cF$ is not a well-rounded convex body in
$\RR^m$ (it has zero volume), so we will have to apply the following standard pre-processing step to transform it into a well-rounded body. First, we find a maximal set $\cI$ of linearly independent
points in $\cF$. We will then embed the polytope $\cF$ into this
lower-dimensional subspace spanned by $\cI$, so that $\cF$ becomes
full-dimensional. In this subspace, $\cF$ is a convex body with a relative interior.  Next,
we apply the transformation of ~\cite{LV06} to transform $\cF$
into a well-rounded body within $\mathrm{Span}(\cI)$.\footnote{See
  Section 5 of~\cite{LV06} for details of the rounding algorithm.} We
will run $\zero$ over the transformed body.
}

\begin{lemma}\label{lem:flowmain}
  Let $\alpha > 0$ be the target accuracy. The instantiation
  $\lop(\cA_F, \alpha)$ computes a toll vector $\hat \tau$ such that
  the induced flow $\hat f = f^*(\hat \tau)$ is $\alpha$-approximately
  optimal in expectation:
  \[
  \Expectation \left[\Psi\left(\hat f \right) \right] \leq \min_{f\in \cF} \Psi(f) + \alpha.
  \]
  The total number of observations we need on the flow behavior is bounded by
  \[
  \tilde O\left(\frac{m^{11.5}}{\alpha^4} \right).
  \]
\end{lemma}
\begin{remark}
Just as with the profit maximization example, if we do not require noise tolerance, then we can improve the dependence on the approximation parameter $\alpha$ to be polylogarithmic. We show how to do this in the appendix.
\end{remark}

\section{The Principal-Agent Problem}
Our general framework applies even when the leader observes only the \emph{noisy feedback} that arises when the follower only approximately maximizes her utility function. This corresponds to \emph{adversarially chosen} noise of bounded magnitude. In this section, we show how to handle the natural setting in which the noise being added need not be bounded, but is well behaved -- specifically has mean 0, and bounded variance. This can be used to model actual noise in an interaction, rather than a failure to exactly maximize a utility function.  As a running example as we work out the details, we will discuss a simple \emph{principal-agent} problem related to our profit-maximization example.

In a \emph{principal-agent} problem, the principal (the leader) defines a contract by which the agent (the follower) will be paid, as a function of work produced by the agent. The key property of principal agent problems is that the agent is not able to deterministically produce work of a given quality. Instead, the agent chooses (and experiences cost as a function of) a level of \emph{effort}, which stochastically maps to the quality of his work. However, the effort chosen by the agent is unobservable to the principal -- only the quality of the finished product.

We consider a simple $d$-dimensional principal-agent problem, in which the result of the agent can be evaluated along $d$ dimensions, each of which might require a different amount of effort. Since the agent knows how effort is stochastically mapped to realizations, we abstract away the agent's choice of an ``effort'' vector, and instead (without loss of generality) view the agent as choosing a ``target contribution'' $x \in C \subseteq \RRP^d$ -- the expected value of the agent's ultimate contribution. The agent experiences some strongly convex cost $c(x)$ for producing a target contribution of $x$, but might nevertheless be incentivized to produce high quality contributions by the contract offered by the principal. However, the contribution that is actually realized (and that the principal observes) is a stochastically perturbed version of $x$: $\tilde x = x + \theta$, where $\theta\in \RR^d$ is a noise
vector sampled from the mean-zero Gaussian distribution $\cN(\mathbf{0}, I)$.

The principal wants to optimize over the set of linear contracts: he will choose a price vector $p\in \RRP^d$, such that in response to the agent's realized contribution $\tilde x$, the agent collects reward $\langle p, \tilde x\rangle$. His goal is to choose a price vector to optimize his expected value for the agent's contribution, minus his own costs.

The agent's strongly convex cost function $c\colon C\to \RRP$
is unknown to the principal. If the principal's contract vector is $p$ and the agent
attempts to contribute $x$, then his utility is
\[
U_a(p, x) = \langle p , (x + \theta) \rangle - c(x),
\]
and his expected utility is just $u_a(p, x) = \Expectation [U_a(p, x)] =
\langle p, x\rangle - c(x)$. Fixing any price $p$, the agent will
attempt to play the \emph{induced} contribution vector: $x^*(p) =
\argmax_{x\in C} \left(\langle p, x \rangle - c(x)\right)$ in order to optimize his
expected utility.

 The principal has value $v_i$ for each unit of contribution in the
 $i$-th dimension, and upon observing the realized contribution $\tilde x$,
 his utility is
\[
u_p(p, \tilde x) = \langle v, \tilde x \rangle - \langle p, \tilde x \rangle = \langle v- p, \tilde x \rangle.
\]
 The principal's goal is to find a price vector $\hat p$ to
 (approximately) maximize his expected utility:
\[
\Expectation [u_p(p, x^*(p) + \theta)] = \Expectation\left[\langle v -
  p, x^*(p) + \theta \rangle\right] = \langle v- p, x^*(p)\rangle.
\]

This is an instantiation of our class of Stackelberg games in which
the principal is the leader with action set $\RRP^d$ and utility
function $\psi(p, x) = \langle v - p, x\rangle$, and the agent is the
follower with action set $C$ and utility function $\phi(p, x) =
\langle p, x\rangle - c(x)$. Indeed, in expectation, it is merely a
``procurement'' version of our profit-maximization example. However,
the crucial difference in this application (causing it to deviate from
the general setting defined in~\Cref{def:problem}) is that the leader
only gets to observe a noisy version of the follower's best response
at each round: $\tilde x = x^*(p) + \theta$. We will adapt the
analysis from~\Cref{sec:reveal} and~\Cref{sec:general} to show that
our algorithm is robust to noisy observations.  We make the following
assumptions, which correspond to the set of assumptions we made in our
previous applications.

\begin{assumption} The following assumptions parallel~\Cref{ass:general} and~\Cref{ass:feasible}.
  \begin{enumerate}
  \item The set of feasible contributions $C \subseteq \RRP^d$ is
    convex, closed, and bounded.  It also contains the unit hypercube,
    $[0,1]^d \subseteq C$ (the agent can simultaneously attempt to contribute at
    least one unit in each dimension) and in particular contains
    $\mathbf{0} \in \RR^d$ (the agent can contribute nothing).
    Lastly, $\|C\|_2 \leq \gamma$;
  \item the agent's cost function $c$ is homogeneous, 1-Lipschitz and
    $\sigma$-strongly convex;
  \item the principal's valuation vector has norm $\|v\|\leq 1$.
  \end{enumerate}
\end{assumption}

\subsection{Inducing the Agent's Contribution Using Noisy Observations}
We will first show that in general, $\LL$ can learn the leader's action which
approximately induces any target follower action $\hat x$ even if
the algorithm only observes noisy perturbed best responses from the
follower. This result holds in full generality, but we illustrate it by using the principal-agent problem.

First, given any target contribution $\hat x$, consider the
following convex program similar to~\Cref{sec:conversion}:
\begin{align}
  &\min_{x \in C} c(x) \qquad \label{eq:costobjective}\\ \mbox{ such
that }\qquad &x_j \geq \hat x_j \mbox{ for every }j\in
[d]\label{eq:contribute}
\end{align}
The Lagrangian of the program is
\[
\cL(x, p) = c(x) + \langle p, x -\hat x\rangle,
\]
and the Lagrangian dual objective function is
\[
g(p) = \min_{x\in C} \cL(x, p).
\]
By the same analysis used in the proof of~\Cref{lem:approxP}, if we find a price vector
$\hat p\in \cP$ such that $g(\hat p) \geq \max_{p\in \cP} g(p) -
\alpha$, then we know that the induced contribution vector $x^*(\hat p)$ satisfies
$\|x^*(\hat p) - \hat x\| \leq \sqrt{2\alpha/\sigma}$. Now we show how
to (approximately) optimize the function $g$ based on the realized
contributions of the agent, which correspond to mean-zero perturbations of the agent's best response.

As shown in~\Cref{lem:grad}, a subgradient of $g$ at price $p$ is
$(x^*(p) - \hat x)$, but now since the principal only observes the
realized contribution vector $\tilde x$, our algorithm does not have access to subgradients. However, we can still obtain an unbiased estimate of the
subgradient: the vector $(\tilde x - \hat x)$ satisfies
$\Expectation\left[ \tilde x - \hat x\right] = (x^*(p) - \hat x)$
because the noise vector is drawn from $\cN(\mathbf{0}, I)$. This is sufficient to allow us to analyze $\LL$ as stochastic gradient descent. The principal does the following: initialize $p^1 = \mathbf{0}$ and at each round
$t \in [T]$, observes a realized contribution vector $\tilde x^t =
x^*(p^t) + \theta^t$ and updates the contract prices as follows:
\[
p^{t+1} = \proj{\cP}\left[ p^t + \eta(\tilde x^t - \hat x) \right],
\]
where each $\theta^t\sim \cN(\mathbf{0}, I)$, $\eta$ is a learning
rate and $\cP = \{p\in \RRP^d \mid \|p\| \leq \sqrt{d}\}$;
Finally, the algorithm outputs the average price vector $\hat p =
1/T\sum_{t=1}^T p^t$.  We use the following standard theorem about the
convergence guarantee for stochastic gradient descent (a more general result can be found
in~\cite{NJLS09}).

\begin{lemma}
\label{thm:noisy-GD}
With probability at least $1 - \beta$, the average vector $\hat p$
output by stochastic gradient descent satisfies
\[
\max_{p\in \cP} g(p) - g(\hat p) \leq
O\left(\frac{\sqrt{d}}{\sqrt{T}}\left( \gamma +
\sqrt{d}\log\left(\frac{Td}{\beta} \right)\right) \right).
\]
\end{lemma}

\begin{algorithm}[h]
  \caption{Learning the price vector from noisy observations:
    $\lpn(\hat x, \eps, \beta)$}
 \label{alg:learnpriceN}
  \begin{algorithmic}

    \STATE{\textbf{Input:} A target contribution $\hat x\in C$,
      target accuracy $\eps$, and confidence parameter $\beta$}
    \INDSTATE{Initialize: restricted price space $\cP=\{p\in \RR_+^{d}
      \mid \|p\|\leq \sqrt{d}\}$
      \[
      p^1_j = 0 \mbox{ for all }j \in [d] \qquad
      T=  \tilde O\left(\frac{d\gamma^2}{\eps^4\sigma^2}\right)\qquad
      \eta= \frac{\sqrt{2}\gamma}{\sqrt{d}\sqrt{T}}
      \]
    }
    \INDSTATE{For $t = 1, \ldots , T$:}
    \INDSTATE[2]{Observe the realized contribution by the agent $\tilde
      x^t = x^*(p^t) + \theta$, where $\theta \sim \cN(\mathbf{0}, I)$}
    \INDSTATE[2]{Update price vector:
      \[
      \tilde{p}^{t+1}_j = p^t_j + \eta \left(\hat x_j - \tilde x^t_j\right) \mbox{ for each }j\in [d],\qquad
p^{t+1} = \proj{\cP} \left[\hat p^{t+1}\right]
      \]
}
    \INDSTATE{\textbf{Output:} $\hat p = 1/T\sum_{t=1}^T p^t.$}
    \end{algorithmic}
  \end{algorithm}

\begin{lemma}
\label{lem:epsclose}
Let $\hat x\in C$ be any target contribution vector. Then, with
probability at least $1 - \beta$, the algorithm $\lpn(\hat x, \eps,
\beta)$ outputs a contract price vector $\hat p$ for the principal such that
the induced contribution vector $x^*(\hat p)$ satisfies
\[
\|\hat x - x^*(\hat p)\|\leq \eps,
\]
and the number observations on the \emph{realized} contributions of
the agent it needs is no more than
\[
T =\tilde O\left(\frac{d\gamma^2}{\eps^4\sigma^2}\right).
\]
\end{lemma}

\subsection{Optimizing the Principal's Utility}
Finally, we show how to optimize the principal's utility by
combining $\lpn$ and $\zero$.

Following from the same analysis of~\Cref{lem:best-price}, we know
that the principal's utility-maximizing price vector to induce
expected contribution $\hat x$ is $\nabla c(\hat x)$. We can then
rewrite the expected utility of the principal as a function of the
attempted contribution of the agent:
\[
u_p(x)  = \langle v - \nabla c(x) , x \rangle.
\]
Since $c$ is a homogeneous and convex function,
by~\Cref{thm:applyEuler}, $u_p$ is a concave function.

Similar to~\Cref{sec:prof}, we will run $\zero$ to optimize over the
interior subset:
\[
C_\delta = (1 - 2\delta)C + \delta \mathbf{1},
\]
so any price vector $\hat p$ given by $\lpn$ is the unique price that
induces the agent's attempted contribution vector $x^*(\hat p)$
(\Cref{lem:intP}). By the same analysis of~\Cref{lem:intApprox}, we
know that there is little loss in principal's utility by restricting
the contribution vectors to $C_\delta$.
\begin{lemma}
\label{lem:intapprox}
The function $u_p\colon C\to \RR$ is 2-Lipschitz, and for any $0<
\delta < 1$,
$$\max_{x\in C} u_p(x) - \max_{x\in C_\delta} u_p(x) \leq
6\delta\gamma.
$$
\end{lemma}

Now we show how to use $\lpn$ to provide an noisy evaluation for $u_p$
at each point of $C_\delta$ (scale of $\delta$ determined in the
analysis). For each $\hat p$ the $\lpn$ returns, the realized
contribution vector we observe is $\tilde x = x^*(\hat p) + \theta$,
so the utility experienced by the principal is
\[
u_p(\hat p, \tilde x) = \langle v - \hat p, \tilde x\rangle.
\]
We first demonstrate that $u_p(\hat p, \tilde x)$ gives an unbiased
estimate for $u_p$, and we can obtain an accurate estimate by taking
the average of a small number realized utilities. In the following,
let constant $a= \ln 2/(2\pi)$.

\begin{lemma}
\label{lem:avgacc}
  Let $x'\in C$ be the contribution vector such that $p' = \nabla c(x')$
  is the unique price vector that induces $x'$. Let noise vectors
  $\theta^1, \ldots , \theta^s\sim\cN(\mathbf{0}, I)$ and $\tilde x^j
  = x' + \theta^j$ for each $j\in [s]$. Then with probability at least
  $1 - \beta$,
  \[
\left |\frac{1}{s}\sum_{j=1}^su_p(\hat p, \tilde x^j) - u_p(x')\right|
\leq \sqrt{\frac{d}{s}} \sqrt{\frac{2}{a}\ln{\frac{2}{\beta}}}.
  \]
\end{lemma}

\begin{proof}
Let $b = v - p'$, then we can write
\begin{align*}
  \frac{1}{s}\sum_{j=1}^su_p(\hat p, \tilde x^j) - u_p(x') &=
  \frac{1}{s}\sum_{j=1}^s\left( \langle b, \tilde x^j \rangle - \langle b, x'\rangle  \right)\\
  &=  \frac{1}{s}\sum_{j=1}^s \langle b, \theta^j \rangle\\
  &= \frac{1}{s} \sum_{j=1}^s \sum_{i=1}^d b_i \theta^j_i
\end{align*}

Note that each $\theta_i^j$ is sampled from the Gaussian distribution
$\cN(0, 1)$, and we use the fact that if $X\sim \cN(0, \sigma_1^2)$
and $Y\sim \cN(0, \sigma_2^2)$ then $(b X +c Y)\sim \cN(0, b^2\sigma_1^2 +
c^2\sigma_2^2)$. We can further derive that $\frac{1}{s} \sum_{j=1}^s
\sum_{i=1}^d b_i \theta^j_i$ is a random variable with distribution
$\cN(0, \|b\|^2/s)$. Then we will use the following fact about
Gaussian tails: let $Y$ be a random variable sampled from distribution
$\cN(0, \iota^2)$ and $a = \ln{2} /(2\pi)$, then for all $\zeta > 0$
\[
\Pr\left[ |Y| > \zeta \right]\leq 2\exp\left( -a \zeta^2/\iota^2\right)
\]
It follows that with probability at least $1-\beta$, we have
\[
\left| \frac{1}{s} \sum_{j=1}^s\sum_{i=1}^d b_i \theta_i^j\right| \leq
\sqrt{\frac{\ln{\frac{2}{\beta}}}{a s}}\|b\|.
\]
Finally, note that we can bound $\|b\| = \|v - p'\| \leq \sqrt{2d}$,
so replacing $\|b\|$ by $\sqrt{2d}$ recovers our bound.
\end{proof}
Now we are ready to give the algorithm to optimize the principal's
utility in~\Cref{alg:optpriceN}.
\begin{algorithm}[h]
  \caption{Learning the price vector to optimize under noisy observations: $\opn(C,
    \alpha, \beta)$}
 \label{alg:optpriceN}
  \begin{algorithmic}
    \STATE{\textbf{Input:} Feasible bundle space $C$, target accuracy $\alpha$, and confidence parameter $\beta$}

    \INDSTATE{Initialize:
      \[
      \eps = \frac{\alpha}{12\gamma + 3d} \qquad
      \delta = 2\eps \qquad
      \alpha' =  3d\eps\qquad
      \beta' = \beta/2T \qquad
      s = \frac{2d\ln{\frac{2}{\beta'}}}{a\eps^2}
      \]

restricted bundle space $C_\delta = (1 -
      2\delta) C + \delta\mathbf{1}$ and number of iterations $T=
      \tilde O(d^{4.5})$ }
    \INDSTATE{For $t = 1, \ldots , T$:}
    \INDSTATE[2]{$\zero(\alpha', C_\delta)$ queries the profit for bundle $x^t$}
    \INDSTATE[2]{Let $ p^t = \lpn(x^t, \eps, \beta')$}
    \INDSTATE[2]{For $j = 1, \ldots s$:}
    \INDSTATE[3]{Principal post price $p^t$}
    \INDSTATE[3]{Let $\tilde x^j(p^t)$ be the realized contribution and experiences utility $u(p^t, \tilde x^j(p^t))$}
    \INDSTATE[2]{Send $\frac{1}{s}\sum_{j=1}^s u(p^t, \tilde x^j(p^t))$ to
      $\zero(\alpha', C_\delta)$ as an approximate evaluation of
      $u_p(x^t)$}
    \INDSTATE{$\hat x = \zero(\alpha', C_\delta)$}
    \INDSTATE{$\hat p = \lp(\hat x, \eps)$}
    \INDSTATE{\textbf{Output:} the last price vector $\hat p$}
    \end{algorithmic}
  \end{algorithm}

\begin{theorem}
\label{thm:noisyaccuracy}  
Let $\alpha > 0$ and $0< \beta < 1/2$.  With probability at least
$1-\beta$, the price vector $\hat p$ output by $\opn(C, \alpha,
\beta)$ satisfies
\[
\Expectation\left[u_p(\hat p, x^*(\hat p))\right] \geq \max_{p\in \cP}
u_p(p, x^*(p)) -\alpha,
\]
and the number of observations on realized contributions is bounded by
\[
\tilde O\left( \frac{d^{9.5}}{\alpha^4}\right).
\]
\end{theorem}

\begin{proof}
  First, by~\Cref{lem:epsclose} and union bound, with probability at
  least $1 - \beta/2$, we have $\|x^t - x^*(p^t)\|\leq \eps$ for all
  $t\in [T]$. We condition on this level of accuracy for the rest of
  the proof. By the same analysis of~\Cref{thm:optProf}, we know that
  each target contribution $x^*(p^t)$ is in the interior
  $\mathrm{Int}_C$, so we have that $u_p(x^*(p^t)) = u_p(p^t ,
  x^*(p^t))$.

To establish the accuracy guarantee, we need to bound two sources of
error. First, we need to bound the error from $\zero$. Note that the
target contribution $x^*(p^t)$ satisfies
\[
|u_p(x^t) - u_p(x^*(p^t))| \leq 2\eps.
\]
By~\Cref{lem:avgacc} and our setting of $s$, we have with probability
at least $1 - \beta'$ that
\[
\left|\frac{1}{s}\sum_{j=1}^s u_p(p^t, \tilde x^j(p^t)) - u_p(x^*(p^t)) \right| \leq \eps.
\]
By union bound, we know such accuracy holds for all $t\in [T]$ with
probability at least $1-\beta/2$. We condition on this level of
accuracy, then the average utility provides an accurate evaluation for
$u_p(x^t)$ at each queried point $x^t$
\[
\left|\frac{1}{s}\sum_{j=1}^s u_p(p^t, \tilde x^j(p^t)) - u_p(x^t) \right| \leq 3\eps.
\]
By~\Cref{lem:zero}, we know that the vector $\hat x$ output by $\zero$ satisfies
\[
\Expectation\left[u_p(\hat x)\right] \geq \max_{x\in C_\delta} u_p(x) - 3d\eps.
\]
Finally, by~\Cref{lem:intapprox} and the value of $\eps$, we also have
\[
\Expectation\left[u_p(\hat x)\right] \geq \max_{x\in C} u_p(x) - (12\eps\gamma + 3d\eps) = \max_{x\in C} u_p(x) - \alpha.
\]
Note $\max_{x\in C} u_p(x) = \max_{p\in \cP} u_p(p, x^*(p))$, so we
have shown the accuracy guarantee. In each iteration, the algorithm
requires $\tilde O\left(\frac{d\gamma^2}{\eps^4\sigma^2} \right)$
noisy observations for running $\lpn$ and $s$ observations for
estimating $u_p(x^*(p^t))$, so the total number of observations is
bounded by
\[
\tilde O\left(d^{4.5} \times \left(  \frac{\gamma^2d(\gamma + d)^4
}{\sigma^2\alpha^4} + \frac{d(\gamma + d)^2}{\alpha^2}
\right) \right) = \tilde O \left(\frac{d^{9.5}}{\alpha^4} \right)
\]
where we hide constants $\sigma, \gamma$ in the last equality.
\end{proof}

\section{Conclusion}
In this paper, we have given algorithms for optimally solving a large
class of Stackelberg games in which the leader has only ``revealed
preferences'' feedback about the follower's utility function, with
applications both to profit maximization from revealed preferences
data, and optimal tolling in congestion games. We believe this is a
very natural model in which to have access to agent utility functions,
and that pursuing this line of work will be fruitful. There are many
interesting directions, but let us highlight one in particular. In our profit maximization application, it would be very natural to
consider a ``Bayesian'' version of our problem. At each round, the
producer sets prices, at which point a new consumer, with valuation
function drawn from an unknown prior, purchases her utility maximizing
bundle. The producer's goal is to find the prices that maximize
her \emph{expected} profit, over draws from the unknown prior. Under
what conditions can we solve this problem efficiently? The main
challenge (and the reason why it likely requires new techniques) is
that the \emph{expected value} of the purchased bundle need not
maximize any well-behaved utility function, even if each individual
consumer is maximizing a concave utility function.


\section*{Acknowledgements} We would like to thank Michael Kearns and Mallesh Pai for stimulating discussions about this work. In particular, we thank Mallesh for helpful discussions about principal-agent problems and conditions under which the profit function of the producer in the revealed preferences problem might be concave. We would also like to thank Tengyuan Liang and Alexander Rakhlin for very helpful discussions about \cite{BLNR15}.

\fi

\bibliographystyle{alpha}

\bibliography{./main.bbl}
\iffull
\appendix

\section{A Routing Game Where Social Cost is Not Convex in The Tolls}
\label{sec:routingexample}

As we stated in the introduction, we can give a simple example of a
routing game in which the function mapping a set of tolls on each of
the edges to the social cost of the equilibrium routing in the routing
game induced by those tolls is not a convex function of the tolls.
The example is related to the canonical examples of \emph{Braess'
Paradox} in routing games.

\usetikzlibrary {positioning}
\definecolor {processblue}{cmyk}{0.96,0,0,0}

\begin{figure}[h!]
\begin{center}
\begin{tikzpicture}[-latex ,auto ,node distance = 2 cm and 4cm ,on grid ,
thick ,
state/.style ={ circle ,top color =white , bottom color = processblue!20 ,
draw,processblue , text=blue , minimum width =1 cm}]
\node[state] (S) {$S$};
\node[state] (A) [above right =of S] {$A$};
\node[state] (T) [below right=of A] {$T$};
\node[state] (B) [below right =of S] {$B$};
\path (S) edge [bend left = 15] node[above left] {$4x/10$} (A);
\path (S) edge [bend right = 15] node[below left] {$1/2$} (B);
\path (B) edge [bend right = 15] node[below right] {$4x/10$} (T);
\path (A) edge [bend left = 15] node[above right] {$1/2$} (T);
\path (A) edge [bend right = 25] node[left] {$1/200 + \tau_{1}$} (B);
\path (A) edge [bend left = 25] node[right] {$\tau_2$} (B);
\end{tikzpicture}
\end{center}
\caption{\footnotesize{A routing game in which the function mapping tolls to social cost of the unique equilibrium routing is not convex.  In this example, there are $n$ players each trying to route $1/n$ units of flow from $S$ to $T$.  There are two edges from $A$ to $B$ with tolls $\tau_1$ and $\tau_2$, and we assume without loss of generality that all other tolls are fixed to $0$.  Each edge is labeled with the latency function indicating the cost of using that edge when the congestion on that edge is $x \in [0,1]$.  Note that the latencies (excluding the tolls) on every edge are bounded in $[0,1]$.}}
\end{figure}
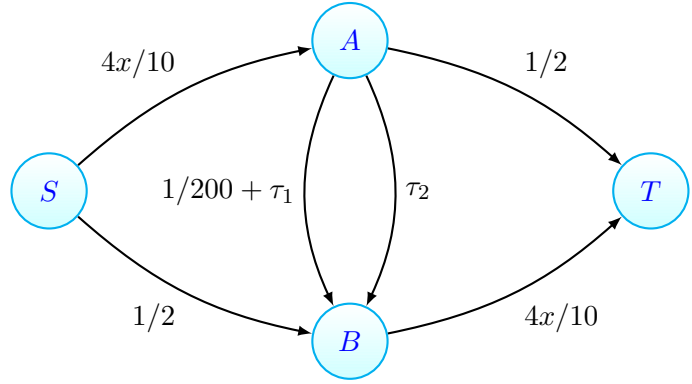

Let $\mathit{SC}(\tau_1,\tau_2)$ be the function that maps a pair of tolls for the two $A \rightarrow B$ edges to the social cost (excluding the tolls) of the equilibrium routing.  For each of the inputs we consider, the equilibrium will be unique, so multiplicity of equilibria is irrelevant.

First, consider the set of tolls $\tau = \tau_1 = \tau_2 = 0$.  It is not hard to verify that the unique equilibrium is for every player to use the route $S \rightarrow A \rightarrow B \rightarrow T$ using the $A \rightarrow B$ edge on the right (with latency $0$).\footnote{Since the graph is a DAG, we can use backwards induction.  From $A$, it can never cost more to go $A \rightarrow B \rightarrow T$ than to go $A \rightarrow T$.  Since one can go from $A$ to $B$ for a cost of $0$, players are indifferent about ending up at node $A$ and node $B$.  Since $S \rightarrow A$ can never cost more than $S \rightarrow B$, and players are indifferent between $A$ and $B$, every player would choose the path $S \rightarrow A \rightarrow B \rightarrow T$ (using the $0$ latency path from $A$ to $B$.}  Each player will experience a total latency of $1$ along their route.  Thus $\mathit{SC}(\tau) = 8n/10$.

Now consider the tolls $\tau'$ in which $\tau_1 = 1, \tau_2 = 2$.  At these tolls, it is not hard to verify that the unique equilibrium is for $n/2$ players to use $S \rightarrow A \rightarrow T$ and half to use $S \rightarrow B \rightarrow T$.\footnote{At these tolls, no player will never use either $A \rightarrow B$ edge.  Thus, they will balance the traffic so that $S \rightarrow A \rightarrow T$ and $S \rightarrow B \rightarrow T$ have equal cost.  By symmetry, half will go through $A$ and half through $B$.}  Every player experiences a total latency of $2/10 + 1/2 = 7/10$.  Thus $\mathit{SC}(t') = 7n/10$.

Finally, consider the convex combination $99\tau/100 + \tau'/100$ in which $\tau_1 = 1/100$ and $\tau_2 = 1/50$.  In this case, the unique equilibrium routing will have every player use the route $S \rightarrow A \rightarrow B \rightarrow T$ but using the $A \rightarrow B$ edge on the left (with latency $1/200$ and latency-plus-toll $3/200$).  To see why, observe that if a player were at $A$, then no matter what the other players are doing, the cheapest path to $T$ is to go $A \rightarrow B \rightarrow T$ using the left edge (note that the right edge has latency-plus-toll $1/50$ whereas the left edge has latency-plus-toll $3/200$).  Thus, the cost of going $B \rightarrow T$ is exactly $1/2$ and the cost of going $A \rightarrow B \rightarrow T$ is exactly $1/2 + 3/200$.  Now, if the player is at $S$, going $S \rightarrow B \rightarrow T$ costs exactly $1$, whereas going $S \rightarrow A \rightarrow B \rightarrow T$ costs \emph{at most} $4/10 + 1/2 + 3/200 = 183/200 < 1/2.$  Thus, every player will choose the path $S \rightarrow A \rightarrow B \rightarrow T$, using the left $A \rightarrow B$ edge.  Every player experiences a total latency of exactly $183/200$.  Thus, $\mathit{SC}(t'') = 183n/200$.

But, since
\begin{align*}
\mathit{SC}(99\tau/100 + \tau'/100)
={} &\frac{183n}{200} \\
>{} &\frac{99}{100} \cdot \frac{8n}{10} + \frac{1}{100} \cdot \frac{7n}{10} \\
={} &\frac{99}{100} \cdot \mathit{SC}(\tau) + \frac{1}{100} \cdot \mathit{SC}(\tau'),
\end{align*}
we can see that the function $\mathit{SC}(\tau)$ is not convex in $\tau$.

\section{Missing Proofs in~\Cref{sec:reveal}}
\label{sec:homoproof}
\begin{lemma}
  Suppose function $v\colon \RRP^d \rightarrow \RRP$ is a concave and
  homogeneous of some degree $k\leq 0$. Then $k\leq 1$.
\end{lemma}
\begin{proof}
First, we show that $v(\mathbf{0}) = 0$. To see this, observe that for
any $b > 1$, we can write $v(\mathbf{0}) = v(b, \mathbf{0}) = b^k
v(\mathbf{0})$. For any $x\in \RRP^d$ such that $x\neq 0$, we have the
following due to the concavity of $v$:
\[
v(x)/2 = \frac{1}{2}\left[v(\mathbf{0}) + v(x) \right] \leq v(x/2) = \left(\frac{1}{2}\right)^k v(x).
\]
This means that $k\leq 1$.
\end{proof}

\subsection{Properties of CES and Cobb-Douglas Utilities}
\label{sec:sconcave}
In this sub-section, we give proofs showing that both CES and
Cobb-Douglas utility functions are strongly concave and H\"{o}lder
continuous in the convex region $(0, H]^d$.

\subsubsection{Constant Elasticity of Substitution (CES)}
Consider valuation functions of the form:
\[
v(x) = \left( \sum_{i=1}^{d} \alpha_i x_i^{\rho} \right)^{\beta},
\]
where $\alpha_i > 0$ for every $i \in [d]$ and $\rho, \beta > 0$
such that $\rho < 1$ and $\beta\rho < 1$.

\begin{theorem}
Let $C$ be a convex set such that there exists some constant that
$H>0$ such that $C\subseteq (0, H]^d$. Then $v$ is $R$-strongly
  concave over the $C$ for some constant $R$ and is $(\left((\max_i
  \alpha_i)d\right)^\beta, \rho\beta)$-H\"{o}lder continuous.
\end{theorem}

\begin{proof}
We will derive the Hessian matrix $\nabla^2 v$ of the function $v$,
and show that there exists some fixed $R > 0$ such that for every
$x\in C$ and every $y\in \RR^d$, we have $y^\intercal
\nabla^2 v(x) y\leq -R\|y\|^2$.  \ju{This imprecise.  It's not enough
  to show negative definite because the Hessian depends on $x \in C$.
  We need to show that there is some fixed $R$ such that for every $x
  \in C$ and every unit $y$, $y^\intercal M y < -R$ where $M$ is the
  Hessian.  Alternatively, the $R$ we use in our analysis is really a
  function of $x$, and the function is actually $(\inf_{x \in C})
  R(x)$-strongly concave. If $C$ is compact and $R$ is continuous then
  the function will attain its minimum, so we can replace the $\inf$
  with a $\min$.}\sw{reworked} First, we have the first partial
derivatives
\begin{align*}
  \frac{\partial v}{\partial x_i} = \beta \left( \sum_{k=1}^d \alpha_k
  x_k^\rho \right)^{(\beta - 1)} \left(\rho \alpha_i\,  x_i^{\rho - 1} \right)
\end{align*}
Now we take the second partial derivatives.  For any $i \neq j$, 
\begin{align*}
  \frac{\partial^2 v}{\partial x_i \partial x_j} = \beta (\beta - 1)
  \left( \sum_{k=1}^d \alpha_k x_k^\rho \right)^{\beta - 2} 
\left(\rho \alpha_j x_j^{\rho - 1} \right)
\left(\rho \alpha_i x_i^{\rho - 1} \right)
\end{align*}
and for any $i$,
\begin{align*}
  \frac{\partial^2 v}{\partial x_i^2} = \beta \left(\beta - 1 \right)
  \left(\sum_{k=1}^d \alpha_k x_k^\rho \right)^{\beta - 2}\left(\rho
  \alpha_i\, x_i^{\rho - 1} \right)^2 + \beta\left(\sum_{k=1}^d
  \alpha_k x_k^\rho \right)^{\beta - 1}(\rho (\rho - 1)\alpha_i x_i^{\rho-2}).
\end{align*}

Recall that the $ij$-th entry of the Hessian matrix is $(\nabla^2
v)_{i,j} = \partial^2 v / \partial x_i \partial x_j$, and we could
write \ju{Need to make it clear that the Hessian is a function of $x$,
  thus what we are really showing is that for $x \in C$, the hessian
  is negative definite.} \ju{I think it's customary to write
  $\partial^2 f / \partial x^2$.  Not $\partial^2 f / \partial^2 x$}

\begin{align*}
y^{\intercal}\left( \nabla^2 v (x)\right) y &= \sum_{i=1}^d\sum_{j=1}^d
\frac{\partial^2 v}{\partial x_i \partial x_j} y_i y_j\\ &=
2\sum_{i\neq j} \frac{\partial^2 v}{\partial x_i \partial x_j} y_i y_j
+ \sum_{i=1}^d \frac{\partial^2 v}{\partial^2 x_i} y_i^2\\ &=
2\beta(\beta - 1) \rho^2 \sum_{i\neq j} \alpha_i\alpha_j
\left(\sum_{k=1}^d \alpha_k x_k^\rho \right)^{\beta - 2} (x_i^{\rho -
  1}y_i) (x_i^{\rho - 1}y_j)\\ &+ \sum_{i=1}^d\left[ \beta(\beta - 1)
  \left(\sum_{k=1}^d \alpha_k x_k^{\rho} \right)^{\beta - 2}
  \left(\rho\alpha_i x_i^{\rho - 1} y_i \right)^2 + \beta \left(
  \sum_{k=1}^d \alpha_k x_k^\rho\right)^{\beta - 1} \left(\rho (\rho -
  1) \alpha_i x_i^{\rho -2} y_i^2 \right) \right]\\ &= \beta (\beta -
1)\rho^2 \left(\sum_{k=1}^d \alpha_k x_k^\rho\right)^{\beta - 2}
\left( \sum_{k=1}^d \alpha_k x_k^{\rho-1} y_k\right)^2 + \beta
\rho(\rho - 1)\left(\sum_{k=1}^d \alpha_k x_k^\rho \right)^{\beta-1}
\left(\sum_{k=1}^d \alpha_k x_k^{\rho-2}y_k^2\right)
\end{align*}

We will first consider the case where $\beta \leq 1$. Then both of the
terms above are non-positive\ju{You're really just using that the first term is $\leq 0$, right?  The first line of the next equation block is exactly the second term (but with some signs flipped).} \ju{Keep indices consistent.  You had a sum with index $i$ above and $k$ below.  I made everything $k$}, and
\begin{align*}
  y^\intercal (\nabla^2 v(x)) y &\leq - \beta \rho (1- \rho) \left(
  \sum_{k=1}^d \alpha_k x_k^\rho \right)^{\beta -1} \left(
  \sum_{k=1}^d \alpha_k x_k^{\rho - 2} y_k^2 \right) \\
  &\leq - \beta \rho (1- \rho) \inf_{x\in C}\left\{\left(
  \sum_{k=1}^d \alpha_k x_k^\rho \right)^{\beta -1} \right\}\, \inf_{k\in [d], x\in C}\{\alpha_k x_k^{\rho - 2}\}
  \left(  \sum_{k=1}^d  y_k^2 \right)\\
  &\leq -  \beta \rho (1- \rho) \left(
  \sum_{k=1}^d \alpha_k H^\rho \right)^{\beta -1} \min_{k}\{\alpha_k \}H^{\rho - 2}\|y\|^2\\
  &= -  \beta \rho (1- \rho) \left(
  \sum_{k=1}^d \alpha_k \right)^{\beta -1} \min_{k}\{\alpha_k \}H^{\rho \beta- 2}\|y\|^2
\end{align*}

\ju{Something is a bit fishy here.  If $x_k = 0$ and $\rho < 2$, we
  would have division by $0$ if we compute $x_{k}^{\rho - 2}$.  I
  think this might have something to do with the fact that when you're
  on the boundary we don't need that every directional second
  derivative is negative, we just need that it's derivative for
  directions that stay in the set $C$.}\sw{I think the proof works
  ok. But maybe we shouldn't claim that our function is differentiable
  everywhere in $C$?}

  Thus, $y^\intercal (\nabla^2 v(x)) y \leq - R \|y\|^2$ for $R = \beta
  \rho (1- \rho) \left( \sum_{k=1}^d \alpha_k \right)^{\beta -1}
  \min_{k}\{\alpha_k \}H^{\rho \beta- 2}$.

Now we consider the case where $\beta > 1$. Since we have assumed that $\rho
\beta < 1$, we also know that $(\beta - 1)\rho < 1 - \rho$.  Let $\kappa =
\frac{(\beta - 1)\rho}{(1 - \rho)}$ and we know that $0 < \kappa < 1$. It follows that

\begin{align*}
  y^{\intercal}\left( \nabla^2 v(x)\right) y 
&= \beta\rho (1- \rho)\left(\sum_{k=1}^d\alpha_k x_k^\rho\right)^{\beta-2}
\left[\kappa \left(\sum_{k=1}^d \alpha_k x_k^{\rho - 1}y_k \right)^2
  - \kappa\left(\sum_{k=1}^d \alpha_k x_k^\rho \right) \left(\sum_{k=1}^d \alpha_k x_k^{\rho -2} y_k^2 \right) \right]\\
&-\beta\rho (1- \rho)\left(\sum_{k=1}^d\alpha_k x_k^\rho\right)^{\beta-2}  \left[ (1- \kappa)\left(\sum_{k=1}^d \alpha_k x_k^\rho \right) \left(\sum_{k=1}^d \alpha_k x_k^{\rho -2} y_k^2 \right) \right]\\
\mbox{(Cauchy-Schwarz)} \quad&\leq 
- \beta\rho (1- \rho) (1 - \kappa)\left(\sum_{k=1}^d\alpha_k x_k^\rho\right)^{\beta-1} \left(\sum_{k=1}^d \alpha_k x_k^{\rho -2} y_k^2 \right)\\
&\leq - \beta\rho (1- \rho) (1 - \kappa) \left(\sum_{k=1}^d \alpha_k^\beta x_k^{\rho\beta -2} y_k^2 \right)\\
&\leq  - \beta\rho (1- \rho) (1 - \kappa) \inf_{k\in [d], x\in C}\left \{\alpha_k^\beta x_k^{\rho\beta -2} \right\} \|y\|^2\\
&\leq  - \beta\rho (1- \rho) (1 - \kappa) \min_{k}\left \{\alpha_k^\beta  \right\} H^{\rho\beta -2} \|y\|^2
\end{align*}
This means, $y^\intercal (\nabla^2 v(x)) y \leq R \|y\|^2$ for $R =
\beta\rho (1- \rho) (1 - \kappa) \min_{k=1}\left \{\alpha_k^\beta
\right\} H^{\rho\beta -2}$. Therefore, we have shown that $v$ is $R$
strongly concave in $C$ for some positive constant $R$.

Next, we will show that the function is H\"{o}lder continuous over
$C$.  Let $x, y\in C$ such that $x\neq y$. Without loss of generality,
assume that $v(x) \geq v(y)$, and let $\eps_i = |x_i - y_i|$ for each
$i\in [d]$. Then we have
\begin{align*}
\left( \sum_i \alpha_i x_i^\rho \right)^\beta - \left( \sum_i \alpha_i y_i^\rho\right)^\beta
&\leq \left( \sum_i \alpha_i |x_i - y_i|^\rho \right)^\beta \\
&\leq \left(\max_i \alpha_i\right)^\beta \cdot \left(\sum_i \eps_i^\rho \right)^\beta\\
&\leq \left(\max_i \alpha_i \right)^\beta\cdot \left(d \|x - y\|_2^\rho \right)^\beta\\
&\leq \left(\max_i \alpha_i \right)^\beta\cdot d^\beta \left\|x - y\right\|_2^{\rho \beta}
\end{align*}
where the first step follows from the sub-additivity. 
This shows that the function is $(\left((\max_i
\alpha_i)d\right)^\beta, \rho\beta)$-H\"{o}lder continuous over $C$,
which completes the proof.
\end{proof}

\subsubsection{Cobb-Douglas}
Consider valuation functions of the form
\[
v(x) = \prod_{i=1}^{d} x_i^{\alpha_i},
\]
where $\alpha_i > 0$ for every $i \in [d]$ and $\sum_{i=1}^{d}
\alpha_i < 1$. 

\begin{theorem}
Let $C$ be a convex set such that there exists some constant that
$H>0$ such that $C\subseteq (0, H]^d$. Then $v$ is $(1, \sum_i
  \alpha_i)$-H\"{o}lder continuous and $R$-strongly concave over $C$
  for some constant $R$.
\end{theorem}

\begin{proof}
Similar to the previous proof, we will show that there exists some
constant $R>0$ such that for every $x\in C$ and every
$y\in \RR^d$, we have $y^\intercal \nabla^2 v(x) y \leq
-R\|y\|^2$. First, we could write down the following first and second
partial derivatives of the function:
\[
\frac{\partial v}{\partial x_i} = \alpha_i x_i^{\alpha_i - 1} \prod_{j\neq i} x_j^{\alpha_j}
\]
\[
\frac{\partial^2 v}{\partial^2 x_i} = \alpha_i (\alpha_i - 1) x_i^{\alpha_i - 2} \prod_{j\neq i} x_j^{\alpha_j}
\]
and for any $i \neq j$,
\[
\frac{\partial^2 v}{\partial x_i \partial x_j} = \alpha_i x_i^{\alpha_i - 1} \alpha_j x_j^{\alpha_j - 1} \prod_{k\neq i, j} x_k^{\alpha_k}
\]
Let $y\in \mathbb{R}^d$, and let $\kappa = \sum_{i=1}^d \alpha_i \in  (0,1)$.
\begin{align*}
  y^\intercal (\nabla^2 v(x)) y & = \sum_{i=1}^d \sum_{j=1}^d y_i \frac{\partial^2 v}{\partial x_i \partial x_j} y_j \\
  &= 2 \sum_{i\neq j} \prod_{k\neq i,j} x_k^{\alpha_k} (\alpha_i x_i^{\alpha_i - 1} y_i) (\alpha_j x_j^{\alpha_j - 1} y_j)
  + \sum_{i=1}^d \alpha_i (\alpha_i - 1) x_i^{\alpha_i - 2} y_i^2\prod_{j\neq i}x_j^{\alpha_j} \\
  &= \left(\prod_{k=1}^d x_k^{\alpha_k} \right)\left[ 2\sum_{i\neq j} (\alpha_i x_i^{-1} y_i) (\alpha_j x_j^{-1}y_j) + \sum_{i=1}^d \alpha_i (\alpha_i - 1) x_i^{-2}y_i^2)
    \right]\\
  &= \left(\prod_{k=1}^d x_k^{\alpha_k} \right)\left[ \left(\sum_{i = 1}^d \alpha_i x_i^{-1} y_i \right)^2 - \sum_{i=1}^d \alpha_i  x_i^{-2}y_i^2 \right]\\
  \mbox{(Cauchy-Schwarz)} \quad &\leq \left(\prod_{k=1}^d x_k^{\alpha_k} \right)\left[ \left(\sum_{i=1}^d \alpha_i \right) \left(\sum_{i = 1}^d \alpha_i x_i^{-2} y_i^2 \right) - \sum_{i=1}^d \alpha_i  x_i^{-2}y_i^2 \right]\\
  &\leq   \left(\prod_{k=1}^d x_k^{\alpha_k} \right)\left[ \kappa\left(\sum_{i = 1}^d \alpha_i x_i^{-2} y_i^2 \right) - \sum_{i=1}^d \alpha_i  x_i^{-2}y_i^2 \right]\\
  &\leq - \left(\prod_{k=1}^d x_k^{\alpha_k} \right)(1 - \kappa)\left(\sum_{i = 1}^d \alpha_i x_i^{-2} y_i^2 \right)\\
  &\leq - H^{(\sum_{k} \alpha_k - 2)} (1 - \kappa) \min_{k} \{\alpha_k\} \|y\|^2
\end{align*}

Therefore, $v$ is $R$-strongly concave in $C$ for
$R=H^{(\sum_{k} \alpha_k - 2)} (1 - \kappa) \min_{k} \{\alpha_k\}$.

Next, we will show that the function is $(1, \sum_i
\alpha_i)$-H\"{o}lder continuous. Let $x, y\in C$ such that $x\neq
y$. Without loss generality, assume that $v(x) \geq v(y)$. We could
write
\begin{align*}
  \prod_{i=1}^{d} x_i^{\alpha_i} -   \prod_{i=1}^{d} y_i^{\alpha_i} &\leq 
\prod_{i=1}^d |x_i - y_i|^{\alpha_i} \\
&\leq \prod_{i=1}^d\|x - y\|_2^{\alpha_i} = \|x-y\|_2^{\sum_i \alpha_i}
\end{align*}
where the first step follows from the sub-additivity of $v$.  This
shows that the function is $(1, \sum_i \alpha_i)$-H\"{o}lder
continuous, which completes the proof.
\end{proof}

\section{Detailed Analysis of~\Cref{sec:general}}
\label{sec:omitproofs}

Just as in~\Cref{sec:conversion}, we start by considering the
following convex program associated with $\hat x$
\begin{align}
&\max_{x \in \cA_F} \phi(x) \qquad \label{eq:objective}\\ \mbox{ such
that }\qquad &x_j \leq \hat x_j \mbox{ for every }j\in
[d]\label{eq:bound}
\end{align}
The Lagrangian of the program
is\[
\cL(x, p) = \phi(x) - \sum_{j = 1}^d p_j (x_j - \hat x_j),
\]
where $p_j$ is the dual variable for each constraint~\eqref{eq:bound}
and the vector $p$ is also interpreted as the action of the
leader. We can interpret the Lagrangian as the payoff function of a
zero-sum game, in which the leader is the minimization player and the
follower is the maximization player. To guarantee the fast convergence of
our later use of gradient descent, we will restrict the leader to play
actions in $\cP$m as defined in~\eqref{eq:P}. It is known that such a
zero-sum game has value $V$: the leader has a minimax strategy $p^*$
such that $\cL(x, p^*) \leq V$ for all $x\in \cA_F$, and the follower
has a maxmin strategy $x^*$ such that $\cL(x^*, p) \geq V$ for all
$p\in \cP$. We first state the following lemma about the \emph{value}
of this zero-sum game.

\begin{lemma}
The value of the Lagrangian zero-sum game is $V = \phi(\hat x)$, that is
\[
\max_{x\in \cA_F} \min_{p\in \cP} \cL(x, p) =  \min_{p\in \cP} \max_{x\in \cA_F} \cL(x, p) = \phi(\hat x)
\]
\end{lemma}
\begin{proof}
The proof is identical to the proof of~\Cref{lem:BPduality}.
\end{proof}

An approximate minimax equilibrium of a zero-sum game is defined as
follows.

\begin{definition}
\label{def:approxeq}
Let $\alpha \geq 0$. A pair of strategies $p\in \cP$ and $x\in \cA_F$
form an $\alpha$-{approximate minimax equilibrium} if
\[
\cL(x, p) \geq \max_{x'\in \cA_F} \cL(x', p) - \alpha \geq V - \alpha \qquad \mbox{ and } \qquad
\cL(x, p) \leq \min_{p'\in \cP} \cL(x, p') + \alpha \leq V + \alpha.
\]
\end{definition}

First observe that fixing a strategy $x$, the maximization player's best response in this zero-sum game
is the same as the follower's best response in the Stackelberg game.

\begin{claim}
Let  $p' \in \cP$ be any action of the leader, then
\[
\arg\max_{x\in \cA_F} \cL(x, p') = x^*(p')
\]
\label{eq:equiv}
\end{claim}

\begin{proof}
We can write
\begin{align*}
\arg\max_{x\in \cA_F} \cL(x, p') &= \arg\max_{x\in \cA_F} \left[\phi(x) - \langle p', x - \hat x\rangle \right]\\
&= \arg\max_{x\in \cA_F}  \left[\phi(x) - \langle p', x \rangle \right]\\
&= \arg\max_{x\in \cA_F}  U_F(p', x)
\end{align*}
The second equality follows from the fact that $\langle p', \hat
x \rangle$ is independent of the choice of $x$.
\end{proof}

Now we show that if a strategy pair $(p', x')$ is an approximate
minimax equilibrium in the zero-sum game, then $p'$ approximately
induces the target action $\hat x$ of the follower. Hence, our task reduces to finding an approximate minimax
equilibrium of the Lagrangian game.
\begin{lemma}
\label{lem:eqclose}
Suppose that a pair of strategies $(p', x')\in \cP\times \cA_F$ forms
an $\alpha$-approximate minimax equilibrium. Then the induced follower
action $x^*(p')$ satisfies $\|\hat x - x^*(p')\| \leq
\sqrt{4\alpha/\sigma}$.
\end{lemma}

\begin{proof}
By the definition of approximate minimax equilibrium, we have
\[
\phi(\hat x) - \alpha \leq \cL(x', p') \leq \phi(\hat x) + \alpha
\]
and also by~\Cref{eq:equiv},
\[
\max_{x\in \cA_F} \cL(x, p') - \alpha = \cL(x^*(p'), p') -\alpha
\leq \cL(x', p') \leq \phi(\hat x) + \alpha
\]

Note that
\[
\cL(\hat x, p') = \phi(\hat x) - \langle p', \hat x - \hat x\rangle = \phi(\hat x).
\]
It follows that $\cL(x^*(p') , p') \leq \cL(\hat x, p') +
2\alpha$. Since $\phi$ is a $\sigma$-strongly concave function, we
know that fixing any leader's action $p$, $\cL$ is also a
$\sigma$-strongly concave function in $x$. By~\Cref{lem:sconvex} and
the above argument, we have
\[
2 \alpha \geq \cL(x^*(p'), p') - \cL(\hat x, p') \geq \frac{\sigma}{2}\|x^*(p') - x\|^2
\]
Hence, we must have $\|x^*(p') - x\| \leq \sqrt{4\alpha/\sigma}$.
\end{proof}

To compute an approximate minimax equilibrium, we will use the
following $T$-round no-regret dynamics: the leader plays online
gradient descent (a ``no-regret'' algorithm), while the follower
selects an $\zeta$-approximate best response every round. In
particular, the leader will produce a sequence of actions $\{p^1,
\ldots , p^T\}$ against the follower's best responses $\{x^1, \ldots,
x^T\}$ such that for each round $t\in [T]$:
\[
p^{t+1} = \Pi_\cP\left[ p^t - \eta \cdot \nabla_p \cL(x^t, p^t)\right]
\qquad \mbox{ and }\qquad
x^t = x'(p^t).
\]
At the end of the dynamics, the leader has \emph{regret} defined as
\[
\cR_L \equiv \frac{1}{T} \sum_{t=1}^T \cL(x^t, p^t) - \frac{1}{T} \min_{p\in \cP}\sum_{t=1}^T\cL(x^t, p).
\]

Now take the average actions for both players for this dynamics:
$\overline x = \frac{1}{T} \sum_{t=1}^T x^t$ and $\overline p
= \frac{1}{T} \sum_{t=1}^T p^t$. A well-known result by~\cite{FS96}
shows that the average plays form an approximate minimax equilibrium.

\begin{theorem}[\cite{FS96}]
\label{thm:fs96}
The average action pair $(\overline p, \overline x)$ forms a $(\cR_L +
\zeta)$-approximate minimax equilibrium of the Lagrangian zero-sum
game.
\end{theorem}

To simulate the no-regret dynamics, we will have the following
$T$-round dynamics between the leader and the follower: in each round
$t$, the leader plays action $p^t$ based the gradient descent update
and observes the induced action $x^*(p^t)$. The gradient of the
Lagrangian $\nabla_p \cL$ can be easily computed based on the
observations. Recall from~\Cref{eq:equiv} that the follower's best
responses to the leader's actions are the same in both the Stackelberg
game and the Lagrangian zero-sum game. This means that the gradient of
the Lagrangian can be obtained as follows
\[
\nabla_p \cL(x^*(p'), p') = \left( \hat x - x^*(p') \right).
\]

In the end, the algorithm will output the average play of the leader
$1/T\sum_{t=1}^T p^t$. The full description of the algorithm $\LL$ is
presented in~\Cref{alg:learnlead}, and before we present the proof
of~\Cref{thm:generalmain}, we first include the no-regret guarantee of
gradient descent.

\begin{lemma}[\cite{Z03}]\label{lem:gd_regret}
  Let $\cD$ be a closed convex set such that $\|\cD\|\leq D$, and let
  $c^1,\dots,c^T$ be a sequence of differentiable, convex functions
  with bounded gradients, that is for every $x\in \cD$,$ ||\nabla
  c^t(x) ||_2 \leq G$.  Let $\eta = \frac{D}{G\sqrt{T}}$ and
  $\omega^{1} = \proj{\cD} \left[\mathbf{0}\right]$ be arbitrary.
  Then if we compute $\omega^{1},\dots,\omega^{T} \in \cD$ based on
  gradient descent $\omega^{t+1} = \proj{\cD}\left[\omega^t -
    \eta\nabla c(\omega^t)\right] $, the regret satisfies
\begin{equation}
\cR \equiv \frac{1}{T} \sum_{t=1}^T c^t(\omega^{t}) - \min_{\omega \in \cD}
\frac{1}{T}\sum_{t=1}^T c^t(\omega) \leq \frac{G D}{\sqrt{T}}
\label{eq:gd_regret}
\end{equation}
\end{lemma}

\begin{proof}[Proof of~\Cref{thm:generalmain}]
We will first bound the regret of the leader in the no-regret
dynamics. Each vector of $\cP$ has norm bounded by
$\sqrt{d}\lambda_F$. Also, since the gradient of the Lagrangian at
point $p'$ is
\[
\nabla_p \cL(x^*(p'), p') = (\hat x - x^*(p')),
\]
we then know that the gradient is bounded the norm of $\cA_F$, which
is $\gamma$. The regret is bounded by
\[
\cR_L = \frac{\sqrt{2d}\lambda_F \gamma}{\sqrt{T}}.
\]
Let $\overline x = 1/T\sum_{t=1}^T x^*(p^t)$ denote the average play
of the follower. It follows from~\Cref{thm:fs96} that $(\hat p,
\overline x)$ forms an $(\cR_L + \zeta)$-approximate minimax
equilibrium, and so by~\Cref{lem:eqclose}, we have
\[
\|\hat x - x^*(\hat p) \| \leq \sqrt{\frac{4(\cR_L + \zeta)}{\sigma}} =
\sqrt{\frac{4\sqrt{2d}\lambda_F \gamma}{\sigma\sqrt{T}} + \frac{4\zeta}{\sigma}}.
\]
Plugging in our choice of $T$, we get $\|x^*(\hat p) - \hat x\|\leq
\eps/2$, and also the total number of observations on the follower is
also $T$. Finally, by strong
concavity of the function $\phi$ and~\Cref{lem:sconvex}, we have that 
\[
\|x'(\hat p) - x^*(\hat p)\| \leq \sqrt{\frac{2\zeta}{\sigma}} \leq \eps/2.
\]
By triangle inequality, we could show that the approximate
best-response satisfies $\|x'(\hat p) - \hat x\|\leq \eps$.
\end{proof}

Finally, we give the proof of~\Cref{thm:genmain}.

\begin{proof}[Proof of~\Cref{thm:genmain}]
Since for each $x^t$, the approximate evaluation based on $x'(p^t)$
satisfies $|\psi(x^t) - \psi(x'(p^t))| \leq \lambda_L \eps$,
by~\Cref{lem:zero} we know that the action $\hat x$ output by
$\zero(d\eps\lambda_L, \cA_F)$ satisfies
\[
\Expectation \left[\psi(\hat x)\right] \geq \max_{x\in \cA_F}\psi(x) - d\eps\lambda_L
 = \max_{p\in \cA_L} U_L(p, x^*(p)) - d\eps\lambda_L.
\]
Finally, we will use $\LL(\hat x, \eps)$ to output a leader action
$\hat p$ such that $\|x^*(\hat p) - \hat x\| \leq \eps$, and so
$\psi(x^*(\hat p)) \geq \psi(\hat x) - \lambda_L\eps$. Therefore, in
the end we guarantee that
\[
\Expectation\left[U_L(\hat p , x^*(\hat p))\right] = \Expectation\left[\psi(\hat p)\right]
 \geq \Expectation\left[\psi(\hat x) \right]
 -\lambda_L \eps \geq \max_{p\in \cA_L} U_L(p, x^*(p)) -
 (d+1)\eps \lambda_L.
\]
Plugging in our choice of $\eps$, we recover the $\alpha$ accuracy
guarantee. Note that in each iteration, the number of observations on
the follower required by the call of $\LL$ is \[
T' = O\left( \frac{d \lambda_F^2\gamma^2}{\eps^4 \sigma^2}\right)
\]
Therefore the total number of observations our algorithm needs is
bounded by
\[
O\left(T'  \times T \right) = \tilde O\left( \frac{d^{5.5} \lambda_F^2 \gamma^2}{\eps^4\sigma^2} \right)
\]
Again plugging in our value of $\eps$, the bound on the number of
observations needed is \[
\tilde O\left(\frac{d^{9.5} \lambda_F^2 \lambda_L^4\gamma^2}{\alpha^4 \sigma^2}\right).\]
Hiding the constants ($\lambda_F, \lambda_L, \gamma$ and $\sigma$), we
recover the bound above.
\end{proof}

\section{Improvement with Ellipsoid in Noiseless Settings}\label{sec:ellip}
In this section, we present variants of the $\lp$ and $\LL$ algorithms
that uses the Ellipsoid algorithm as a first-order optimization
method. In particular, this will allow us to improve the dependence of
the query complexity on the target accuracy parameter $\alpha$. For
the technique we give in this section, the number of observations of
the follower's actions will have a poly-logarithmic dependence on
$1/\alpha$ instead of a polynomial one.  We also improve the
polynomial dependence on the dimension.

\subsection{The Ellipsoid Algorithm}
We will briefly describe the ellipsoid method without going into full
detail. Let $\cP\subset \RR^d$ be a convex body, and $f\colon
\cP\rightarrow [-B, B]$ be a continuous and convex function. Let $r ,
R > 0$ be such that the set $\cP$ is contained in an Euclidean ball of
radius $R$ and it contains a ball of radius $r$. The \emph{ellipsoid}
algorithm solves the following problem: $\min_{p\in \cP} f(p)$.

The algorithm requires access to a \emph{separation oracle} for $\cP$:
given any $p\in \RR^d$, it either outputs that $x$ is a member of
$\cP$, or if $p\not\in \cP$ then it outputs a separating hyperplane
between $p$ and $\cP$. It also requires access to a \emph{first-order
  oracle}: given any $p\in \RR^d$, it outputs a subgradient $w\in
\partial f(p)$. The algorithm maintains an ellipsoid $E^t$ with center
$c^t$ in $\RR^d$ over rounds, and in each round $t$ does the
following:
\begin{enumerate}
\item If the center of the ellipsoid $c^t\not\in \cP$, it calls the
  separation oracle to obtain a separating hyperplane $w_t\in \RR^d$
  such that $\cP\subset \{p\colon (p - c_t)^\intercal w_t\leq 0\}$;
  otherwise it calls the first-order oracle at $c_t$ to obtain $w_t\in
  \partial f(c_t)$.
\item Obtain a new ellipsoid $E^{t+1}$ with center $c_{t+1}$ based on
  the ellipsoid $E^t$ and vector $w_t$. (See e.g.~\cite{Bubeck14} for
  details.) We will treat this ellipsoid update as a black-box step and write it
as a function $\ellip(E, c, w)$ that takes an ellipsoid $E$ along with
its center $c$ and also a vector $w$ as input and returns a new
ellipsoid $E'$ with its center $c'$ as output.

\end{enumerate}

The sequence of ellipsoid centers $\{c_t\}$ produced by the algorithm
has the following guarantee.

\begin{theorem}[see e.g.~\cite{Bubeck14}]\label{thm:ellip}
For $T \geq 2d^2\log(R/r)$, the ellipsoid algorithm satisfies $\{c_1,
\ldots , c_T \} \cap \cP\neq \emptyset$ and
\[
\min_{c\in \{c_1, \ldots , c_T\}\cap\cP}f(c) - \min_{p\in \cP} f(p)\leq
\frac{2BR}{r}\exp\left( - \frac{T}{2d^2}\right).
\]
\end{theorem}

In other words, with at most $T = O\left(d^2
\log\left(\frac{BR}{r\eps} \right) \right)$ calls to the first-order
oracles, the ellipsoid algorithm finds a point $p\in \left(\{c_t\}\cap
\cP\right)$ that is $\eps$-optimal for the function $f$.

\subsection{Learning Prices with Ellipsoid}
We will first revisit the problem in~\Cref{sec:conversion} and give an
ellipsoid-based variant of $\lp$ that (when there is no noise) obtains
a better query complexity for $\lp$. Recall that we are interested in
computing a price vector $\hat p\in \cP$ such that the induced bundle
$x^*(\hat p)$ is close to the target bundle $\hat x$.

Recall from~\Cref{lem:approxP} that it is sufficient to find an
approximately optimal solution to the Lagrangian dual function $g(p) =
\argmax_{x\in C} \cL(x, p)$. We will use the ellipsoid algorithm to
find such a solution. Note that the feasible region for prices is
given explicitly: $\cP=\{p\in \RR_+^{d} \mid \|p\|\leq \sqrt{d}L\}$,
so a separation oracle for $\cP$ is easy to implement. Furthermore, we
know from~\Cref{lem:grad} that for any $p$, we can obtain a gradient
based on the buyer's purchased bundle: $(\hat x - x^*(p))\in \partial
g(p)$. This means we have both the separation oracle and first-order
oracle necessary for running ellipsoid. The algorithm $\lpe$ is
presented in~\Cref{alg:learnpriceE}.

\begin{algorithm}[h]
  \caption{Learning the price vector to induce a target bundle:
    $\lpe(\hat x, \eps)$}
 \label{alg:learnpriceE}
  \begin{algorithmic}
    \STATE{\textbf{Input:} A target bundle $\hat x\in \mathrm{Int}_C$, and target accuracy $\eps$}
    \INDSTATE{Initialize:
 restricted price space $\cP=\{p\in  \RR_+^{d} \mid \|p\|\leq \sqrt{d}L\}$
 \[
c^1 = \mathbf{0} \qquad E^1 = \{p\in \RR^d \mid \|p\| \leq
\sqrt{d}L\} \qquad T = 100d^2 \ln\left( \frac{d \lambdav \gamma}{\eps \sigma} \right) \qquad       L = \left(\lambdav\right)^{1/\beta}\,\left( \frac{4}{\eps^2\sigma} \right)^{(1-\beta)/\beta} \qquad
 \]
    }
    \INDSTATE{For $t = 1, \ldots , T$:}
    \INDSTATE[2]{while $c^t \not\in \cP$ then let obtain a separating hyperplane $w$ and let  $(E^t, c^t) \leftarrow \ellip(E^t, c^t, w)$}
    \INDSTATE[2]{Let $p^t = c^t$}
    \INDSTATE[2]{Observe the purchased bundle by the consumer $x^*(p^t)$}
    \INDSTATE[2]{Update the ellipsoid  $(E^{t+1}, c^{t+1}) \leftarrow \ellip(E^t, c^t, (\hat x - x^*(p^t)))$:
}
    \INDSTATE{\textbf{Output:} $\hat p = \argmin_{p\in \{p^1, \ldots , p^T\}} \|\hat x - x^*(p)\|$}
    \end{algorithmic}
  \end{algorithm}

\begin{theorem}
  Let $\hat x\in \mathrm{Int}_C$ be a target bundle and $\eps >
  0$. Then $\lpe(\hat x, \eps)$ outputs a price vector $\hat p$ such
  that the induced bundle satisfies $\|\hat x - x^*(\hat p)\| \leq
  \eps$ and the number of observations it needs is no more than
  \[
T = O\left( d^2 \ln\left(\frac{d \lambdav \, \gamma}{\eps \sigma} \right)\right)
  \]
\end{theorem}

\begin{proof}
By~\Cref{lem:approxP}, it suffices to show that there exists a price
vector $\hat p\in \{p^t\}_{t=1}^T$ such that
\[
g(\hat p) \leq \min_{p\in \cP} g(p) + \frac{\eps^2\sigma}{4}.
\]
We will show this through the accuracy guarantee of ellipsoid.  Note
that the set $\cP$ is contained in a ball of radius
$2\sqrt{d}\lambdav$ and contains a ball of radius $\sqrt{d}\lambdav$.
Furthermore, the Lagrangian dual objective value is also bounded: for
any $p\in \cP$:
\begin{align*}
  |g(p)| &= |\max_{x\in C} v(x) - \langle p, x - \hat x\rangle| \\
  &\leq \max_{x\in C} v(x) + |\langle p, x^*(p)- \hat x \rangle|\\
  &\leq \lambdav \, \gamma + \|p\|\cdot  \|x^*(p) - \hat x\|\\
  &\leq \lambdav \, \gamma + \sqrt{2d}\lambdav\,\gamma\\
  &\leq 2\sqrt{d}\lambdav \, \gamma
\end{align*}
By~\Cref{thm:ellip}, the following holds
\[
\min_{p\in \{p^1, \ldots, p^T\}} g(p) - \min_{p'\in \cP} g(p') \leq \frac{\eps^2\sigma}{4}.
\]

By~\Cref{lem:approxP}, there exists some $p\in \{p^1,\ldots , p^T\}$
such that the resulting bundle $x^*(p)$ satisfies that $\|\hat x -
x^*(p)\| \leq \eps$. Since we are selecting $\hat p$ as $\hat p =
\argmin_{p\in \{p^1, \ldots , p^T\}} \|\hat x - x^*(p)\|$, we must
have $\|\hat x - x^*(\hat p)\|\leq \eps$.
\end{proof}

Now we could use $\lpe$ to replace $\lp$ in the algorithm $\op$ as a
sub-routine to induce target bundles. The following result follows
from the same proof of~\Cref{thm:optProf}.

\begin{theorem}
  Let $\alpha > 0$ be the target accuracy. If we replace $\lp$ by
  $\lpe$ in our instantiation of $\op(C, \alpha)$, then the output
  price vector $\hat p$ has expected profit
  \[
  \Ex{}{[r(\hat p)]} \geq \max_{p\in \RR_+^d} r(p) - \alpha,
  \]
  the number of times it calls $\lpe$ is bounded by
  $\left(d^{4.5}\cdot \polylog(d, 1/\alpha)\right)$, and the
  observations the algorithm requires from the consumer is at most
  \[
  d^{6.5} \polylog\left(\lambdav, \gamma, \frac{1}{\alpha}, \frac{1}{\sigma} \right).
  \]
\end{theorem}

\subsection{Learning Tolls with Ellipsoid}
We will also revisit the problem in~\Cref{sec:flow}. We give a similar
ellipsoid-based algorithm to induce target
flow. See~\Cref{alg:learnpriceT}.

\begin{algorithm}[h]
  \caption{Learning the toll vector to induce a target flow:
    $\lpf(\hat f, \eps)$}
 \label{alg:learnpriceT}
  \begin{algorithmic}
    \STATE{\textbf{Input:} A target flow $\hat f\in \cF$, and target accuracy $\eps$}
    \INDSTATE{Initialize:
 restricted toll space $\cP=\{p\in  \RR_+^{d} \mid \|p\|\leq m\}$
 \[
c^1 = \mathbf{0} \qquad E^1 = \{x\in \RR^d \mid \|x\| \leq m\} \qquad
T = 100\left( m^2 \ln\left(\frac{m}{\eps \sigma} \right)\right)
 \]
    }
    \INDSTATE{For $t = 1, \ldots , T$:}
    \INDSTATE[2]{while $c^t \not\in \cP$ then let obtain a separating hyperplane $w$ and let  $(E^t, c^t) \leftarrow \ellip(E^t, c^t, w)$}
    \INDSTATE[2]{Let $\tau^t = c^t$}
    \INDSTATE[2]{Observe the induced flow $f^*(\tau^t)$}
    \INDSTATE[2]{Update the ellipsoid  $(E^{t+1}, c^{t+1}) \leftarrow \ellip(E^t, c^t, -(\hat f - f^*(p^t)))$:
}
    \INDSTATE{\textbf{Output:} $\hat \tau = \argmin_{\tau\in \{\tau^1, \ldots , \tau^T\}} \|\hat f - f^*(\tau)\|$}
    \end{algorithmic}
  \end{algorithm}

\begin{theorem}
    Let $\hat f\in \cF$ be a target bundle and $\eps > 0$. Then
    $\lpf(\hat f, \eps)$ outputs a toll vector $\hat \tau$ such that
    the induced flow satisfies $\|\hat f - f^*(\hat \tau)\| \leq \eps$
    and the number of observations it needs is no more than
  \[
T = O\left( m^2 \ln\left(\frac{m}{\eps \sigma} \right)\right)
  \]
\end{theorem}

\begin{proof}
Let function $g$ be defined as
  \[
  g(\tau) = \min_{f\in \cF} \Phi(f) + \langle \tau  , f - \hat f\rangle.
  \]

It suffices to show that there exists some $\tau' \in \{\tau^1, \ldots,
\tau^T\}$ such that $g(\tau') \geq \min_{\tau \in \cP } g(\tau) - \frac{\eps^2 \sigma}{2}$.
Before we instantiate the accuracy theorem of ellipsoid, note that the
set $\cP$ is contained in a ball of radius $m$ and contains a ball of
radius $m/2$, and also that the value of $g(\tau)$ is bounded for any
$\tau \in \cP$:
\begin{align*}
  |g(\tau)| &= |\min_{f\in \cF} \Phi(f) + \langle \tau , f - \hat f\rangle|\\
  &\leq \max_{f\in \cF} \Phi(f) + \max_{f\in \cF} \| f - \hat f\| \|\tau\|\\
  &\leq m + \sqrt{2m} m \leq 2\sqrt{m^3}
\end{align*}
Given that $T = 4 m^2 \ln(m/\eps\sigma)$, we know by~\Cref{thm:ellip}
that \[
\max_{\tau' \in \{\tau^1, \ldots , \tau^T\}} g(\tau') - \max_{\tau\in \cP} g(\tau) \geq \frac{\eps^2\sigma}{2}
\]
Therefore, the output toll vector satisfies
\[
\|\hat f - f^*(\hat \tau)\| \leq \eps.
\]
This completes our proof.
\end{proof}

Finally, with this convergence bound, we could also improve the result
of~\Cref{lem:flowmain}.

\begin{theorem}
  Let $\alpha > 0$ be the target accuracy. If we replace $\LL$ by
  $\lpf$ in the instantiation of $\lop(\cA_F, \alpha)$, then the
  output toll vector $\hat \tau$ and its the induced flow $\hat f =
  f^*(\hat \tau)$ is $\alpha$-approximately optimal in expectation:
  \[
  \Expectation \left[\Psi\left(\hat f \right) \right] \leq \min_{f\in \cF} \Psi(f) + \alpha.
  \]
  The number of times it calls $\lpf$ is bounded by $\left(m^{4.5}
  \cdot \polylog(m, 1/\alpha)\right)$, and so the total number of
  observations we need on the flow behavior is bounded by
  \[
   m^{6.5} \polylog\left(\lambda , \frac{1}{\alpha}, \frac{1}{\sigma} \right).
  \]
\end{theorem}

\fi

\end{document}